\documentclass{amsart}
\usepackage{fullpage,graphicx,subfigure,mathpazo,color}
\usepackage{amsmath,amscd,tikz,mathrsfs}
\usepackage[normalem]{ulem}
\usepackage{amsmath}
\usepackage{epstopdf}
\usepackage{tikz}
\usepackage{setspace}
\usepackage{hyperref}
\newcommand{\ii}{\mathrm{i}}

\newtheorem{thm}{Theorem}

\newtheorem{prop}{Proposition}

\usepackage{graphicx}      
\usepackage{titlesec}

\titleformat{\section}{\centering\LARGE\bfseries}{\thesection}{1em}{}
\titleformat{\subsection}{\Large\bfseries}{\thesubsection}{1em}{}
\begin{document}

\title{Rogue waves and their patterns for the coupled Fokas-Lenells equations}

\author{Liming Ling}
\address{School of Mathematics, South China University of Technology, Guangzhou, China 510641}
\email{linglm@scut.edu.cn}
\author{Huajie Su}
\address{School of Mathematics, South China University of Technology, Guangzhou, China 510641}
\email{zhangxiaoen@scut.edu.cn}

\begin{abstract}


In this work, we explore the rogue wave patterns in the coupled Fokas-Lenells equation by using the Darboux transformation. 
We demonstrate that when one of the internal parameters is large enough, 
the general high-order rogue wave solutions generated at a branch point of multiplicity three can be decomposed into some
first-order outer rogue waves and a lower-order inner rogue wave. 
Remarkably, the positions and the orders of these outer and inner rogue waves are intimately related to Okamoto polynomial hierarchies.


\end{abstract}

\date{\today}

\maketitle

\section{Introduction}\label{sec-introduction}
Integrable equations, such as 
the classical nonlinear Schr\"odinger (NLS) equation \cite{zakharov_exact_1972}, 
derivative-type NLS equation \cite{kaup_exact_1978, chen_integrability_1979, gerdjikov_quadratic_1982}, 
and other integrable equations play a crucial role in describing nonlinear wave fields. As we know, 
the NLS equation is an applicable model to describe the picosecond short pulse, 
while it is not effective for the subpicosecond or femtosecond pulse. 
In this case, we need to consider the high-order nonlinear effects. 
In the 1980s, Hasegawa and Kodama proposed the high-order NLS equation, 
from which several integrable models can be derived, such as the Hirota equation, 
derivative NLS equation, and Sasa-Satsuma equation. 
In 2009, after recalling certain aspects of the standard derivation of the NLS equation in nonlinear fiber optics, 
Lenells \cite{lenells_exactly_2008} derived the following integrable model
\begin{equation}
\ii u_t-\nu u_{tx}+\gamma u_{xx}+\rho |u|^2(u+\ii \nu u_x)=0
\end{equation}
when taking into account certain terms that are normally ignored. This model was first derived by 
Fokas with the aid of bi-Hamitonian methods \cite{fokas_class_1995}, so we named this model as Fokas-Lenells (FL) equation. 
After applying a gauge and coordinate transformation, the equation mentioned above can be reduced to \cite{lenells_exactly_2008, lenells_novel_2009}
\begin{equation}\label{FL}
    u_{xt}+u+\ii |u|^2 u_x=0.
\end{equation}
Fokas and Lenells provided contributions by deriving the Hamiltonian structure 
and inverse scattering method (ISM) of the integrable FL equation in \cite{lenells_novel_2009}. 
Since then, several distinct types of solutions to the FL equation have been constructed using different techniques. 
The rogue wave solutions were derived from \cite{xu_n-order_2012, chen_peregrine_2014}, 
the dark soliton solutions were constructed using the Hirota bilinear method \cite{matsuno_direct_2012}, 
and the algebraic geometry solutions were constructed by Zhao et al. \cite{zhao_algebro-geometric_2013-1}. 

However, in the birefringent optical fiber systems, two wave packets of different carrier frequencies need to be considered. 
The corresponding coupled Fokas-Lenells (CFL) system
which is given in \cite{chen_peregrine_2018,ye_general_2019}
\begin{equation}\label{CFL system}
    \begin{split}
        {{\rm i}} D_{\xi} q_{1,\tau}+\frac{-\eta}{2} q_{1,\xi\xi} + (2|q_{1}|^{2}+\sigma|q_{2}|^2) D_{\xi} q_{1} + \sigma q_{1}q_{2}^{*} (D_{\xi} q_{2})=0, \\
        {{\rm i}} D_{\xi} q_{2,\tau}+\frac{-\eta}{2} q_{2,\xi\xi} + (2\sigma|q_{2}|^{2}+|q_{1}|^2) D_{\xi} q_{1} + q_{2}q_{1}^{*} (D_{\xi} q_{2})=0,
    \end{split}
\end{equation}
can be utilized to describe the propagation of ultrashort optical pulses \cite{matveev_darboux_1991,lenells_novel_2009,baronio_observation_2018}
in the study of ultrafast optics and hydrodynamics, 
where $\eta=\pm 1$ is the type of dispersion with $\sigma=\pm 1 $, $D_{\xi}=1+{\rm i}\nu \frac{\partial}{\partial \xi}$
is a differential operator and $\nu$ is the permutation of the Manakov system. 
The CFL system is also a generalization of the Manakov system that takes into account more physical effects than the latter \cite{chen_peregrine_2014}. 
Manakov system includes the terms of group-velocity dispersion and self- and cross-phase modulation. 
Additionally, the CFL system takes into account the effects of space-time coupling \cite{boyd_nonlinear_2020} 
and self-deepening \cite{moses_controllable_2006}. 
These terms are obtained by considering the slowly varying envelope approximation in \cite{lenells_novel_2009}. 

The CFL system \eqref{CFL system}, with $\xi = \nu \zeta$ and the transformation
\begin{equation}
    \zeta=\eta x- \eta t, \quad \tau= -2 \nu^{2} t, \quad q_{i}=\frac{{\rm i}}{2\nu}{\rm e}^{\eta {\rm i}x+\eta {\rm i}t}u_{i},
\end{equation}
yields the CFL equations
\begin{equation}\label{CFL}
    \begin{split}
        u_{1,xt}+u_1+{\rm i}(|u_1|^2+\frac{1}{2}\sigma |u_2|^2)u_{1,x}+\frac{{\rm i}}{2}\sigma u_1u_2^*u_{2,x}=&0,  \\
        u_{2,xt}+u_2+{\rm i}(\sigma |u_2|^2+\frac{1}{2}|u_1|^2)u_{2,x}+\frac{{\rm i}}{2}u_2u_1^*u_{1,x}=&0,
    \end{split}
\end{equation}
which were initially proposed by Guo and Ling \cite{guo_riemann-hilbert_2012} using the matrix generalization of the Lax pair. 
Ling, Feng, and Zhu delved into the integrability of the CFL equations in \cite{ling_general_2018} 
and constructed multi-Hamiltonian structures using the Tu scheme. 
The Lax pair of the CFL equations is
\begin{equation}\label{CFL-lax}
    \begin{split}
        \mathbf{\Phi}_x=&\mathbf{U}(x,t;\lambda)\mathbf{\Phi}, \qquad \mathbf{U}(x,t;\lambda)={{\rm i}}\lambda^{-2}\sigma_3+\lambda^{-1}\mathbf{Q}_x, \\
        \mathbf{\Phi}_t=&\mathbf{V}(x,t;\lambda)\mathbf{\Phi}, \qquad \mathbf{V}(x,t;\lambda)={{\rm i}}\left(\frac{1}{4}\lambda^2\sigma_3+\frac{1}{2}\sigma_3(\mathbf{Q}^2-\lambda \mathbf{Q})\right),
    \end{split}
\end{equation}
where
\begin{equation*}
    \sigma_3=\begin{pmatrix}
        1&0&0\\
        0&-1&0\\
        0&0&-1
    \end{pmatrix}
        ,\qquad
    \mathbf{Q}=
    \begin{pmatrix}
        0 & v_1&\sigma v_2 \\
        u_1 &0&0 \\
        u_2 &0&0
    \end{pmatrix}.
\end{equation*}
The zero curvature equation 
$\mathbf{U}_t-\mathbf{V}_x+[\mathbf{U},\mathbf{V}]=0$ ($[\mathbf{U},\mathbf{V}]\equiv \mathbf{U}\mathbf{V}-\mathbf{V}\mathbf{U}$ 
is the commutator) for Lax pair \eqref{CFL-lax} yields the CFL equations
with the symmetric condition $v_{i}=u_{i}^{*},i=1,2$ (the superscript $^{*}$ denotes the complex conjugate). 

Some investigations have already been carried out on the CFL equations, including 
the Riemann-Hilbert approach \cite{kang2018multi} and modulation instability \cite{yue_modulation_2021}. 
On the other hand, due to the integrability, 
we can construct different types of exact solutions for the CFL equations utilizing the methods of integrable systems. 
In 2017, Zhang et al. constructed the solitons, breathers, and rogue waves via the Darboux transformation of 
the integrable CFL equations \cite{zhang_solitons_2017}. 
In 2018, Ling et al. utilized the generalized Darboux transformation to obtain general soliton solutions 
\cite{ling_general_2018}, such as bright solitons, bright-dark solitons, and others. 
The general rogue wave solutions were constructed by Ye. et al. \cite{ye_general_2019} in 2019, and some localized waves were constructed
by Yue et al. \cite{yue_modulation_2021} in 2021. 
Drawing upon the rogue wave solutions, 
it has been demonstrated that lower-order rogue waves can have special patterns 
as evidenced in various graphs \cite{kang2018multi,ye_general_2019,yue_modulation_2021}. 

More specifically, for the CFL equations, using the Darboux transformation, 
we can construct rogue wave solutions \cite{ye_general_2019} at the branch points of multiplicity two and three on the 
Riemann surface which is given by the spectral characteristic polynomial. 
In these two cases, Ye et al. \cite{ye_general_2019} presented figures illustrating 
first-order and second-order rogue waves, showcasing their doublet, triplet, quartet, and sextet states. 
The question naturally arises as to how to study these particular patterns for high-order rogue waves. 

Recently, the studies of rogue wave patterns become popular in the field of rogue waves, which can be used to predict higher-order rogue wave events 
and recognize their decomposition mechanism. 
The roots of special polynomials have been found to be closely associated with rogue wave patterns in various equations, 
as evidenced by previous studies.
In 2021, Yang et al. explored the rogue wave patterns of the NLS equation that corresponds to the Yablonskii-Vorob'ev hierarchies in \cite{yang_rogue_2021-1}. 
In their subsequent work \cite{yang_rogue_2022} in 2023, they examined the rogue wave patterns of the Manakov equations 
and the three-wave resonant interaction equation associated with Okamoto polynomial hierarchies. 
In \cite{zhang_rogue_2022}, Zhang et al. demonstrated that the rogue wave patterns of the vector NLS equation are associated with 
generalized Wronskian-Hermite polynomials. 
To the best of our knowledge, there are no studies on the rogue wave patterns for the CFL equations. 
The main contribution of this work is to study the patterns of rogue waves generated at the branch points 
of multiplicity three \cite{ling_general_2018,ye_general_2019}. 

Actually, the patterns of rogue wave solutions generated at branch points of different multiplicity are associated with different polynomials. 
The case of multiplicity two is associated with the Yablonskii-Vorob’ev polynomial hierarchies. 
The deep-going analysis for the rogue wave solutions generated by the case of multiplicity two needs to be given separately. 
In this work, we concentrate on the case of multiplicity three and analyze the rogue wave patterns, 
which are associated with Okamoto polynomial hierarchies \cite{yang_rogue_2022}. 

In contrast to previous studies on rogue wave patterns \cite{yang_rogue_2021-1,yang_rogue_2022,zhang_rogue_2022}, 
we utilize the Lax pair and the Darboux transformation to construct rogue wave solutions. 
Considering that our research is rooted in the integrability of the CFL equations, 
it is conceivable that a similar methodology can be applied to other general integrable systems, 
enabling the derivation and analysis of rogue waves and their associated patterns.


We organize this work as follows. 
In Section \ref{sec-pre}, we introduce Okamoto polynomial hierarchies and the Darboux transformation for the Lax pair. 
In Section \ref{sec-solution}, we introduce the plane wave solutions for the CFL equations and study the branch points of Riemann surfaces given by
the spectral characteristic polynomial. At the branch point of multiplicity three, we construct high-order rogue wave solutions. 
In Section \ref{sec-pattern}, we analyze the patterns of the rogue wave solutions generated at the branch point of multiplicity three. 
By utilizing the root structures of Okamoto polynomial hierarchies, the rogue wave patterns have two parts: 
the outer region and the inner region. 
We decompose the rogue wave solutions into some first-order rogue wave solutions in the outer region 
and a lower-order rogue wave solution in the inner region.

\section{Preliminaries}\label{sec-pre}

To initiate our analysis of rogue wave patterns in the CFL equations, 
we will provide some preliminaries. 
This section will cover Okamoto polynomial hierarchies and the Darboux transformations 
for the Lax pair \eqref{CFL-lax}.

Okamoto polynomial hierarchies, as outlined in the study by Yang et al. \cite{yang_rogue_2022}, 
play a crucial role in the analysis of rogue wave patterns. 
The rogue wave patterns indicate the specific positions of rogue waves when one of the internal parameters is sufficiently large. 
Furthermore, the Darboux transformation is a powerful tool for constructing solitonic solutions \cite{ling_general_2018,ye_general_2019}, 
as it enables us to derive the rogue wave solutions that we aim to investigate.
\subsection{Okamoto polynomial hierarchies}

Okamoto polynomial hierarchies \cite{yang_rogue_2022} 
are a generalization of the Okamoto polynomials \cite{okamoto_studies_1986}. 
Okamoto demonstrated that the logarithmic derivative of the Okamoto polynomials 
yields rational solutions to the Painlev\'e IV equation. 
Lateer, Kajiwara, and Ohta discovered the determinant representation of Okamoto polynomials 
using Schur polynomials \cite{kajiwara_determinant_1998}. 
Building upon this discovery, 
the determinant representation of Okamoto polynomials can be generalized to define its hierarchies \cite{yang_rogue_2022}. 

Before introducing Okamoto polynomial hierarchies, it is necessary to define Schur polynomials. 
Given an infinite dimensional vector 
$\mathbf{x}=(x_{1},x_{2},\cdots)\in \mathbb{C}^{\infty}$, the Schur polynomials $S_{n}$ are defined by
\begin{equation}
    \sum_{n=0}^{\infty}S_{n}(\mathbf{x})\epsilon^{n}=\exp({\sum_{n=1}^{\infty}x_{n}\epsilon^{n}}). 
\end{equation}
We also define $S_n(\mathbf{x})=0$ for $n\leq -1$. 
In order to analyze the rogue wave patterns, we introduce the following propositions regarding Schur polynomials.
\begin{prop}\label{prop-Ok1}
    For any complex $\eta\ne 0$ and infinite dimensional vector $\mathbf{x}=(x_{1},x_{2},\cdots)\in \mathbb{C}^{\infty}$, we have
    \begin{equation}
        S_{n}(x_{1},x_{2},\cdots)=\eta^{n}S_{n}(x_{1}\eta^{-1},x_{2}\eta^{-2},\cdots).
    \end{equation}
\end{prop}
\begin{prop}\label{prop-OK2}
    Given $\mathbf{x}=(x_{1},x_{2},\cdots)\in \mathbb{C}^{\infty}$ and an integer $k\geq 2$. 
    If $x_{i}=\mathcal{O}(\eta),\forall i\ne k$ and 
    $x_{k}=\mathcal{O}(\eta^{k})$, for $n\geq 2$, we have the asymptotic expansion  
    \begin{equation}
        S_{n}(\mathbf{x})=S_{n}(\mathbf{v})+\begin{cases}
            \mathcal{O}(\eta^{n-1}), \quad k\geq 3,\\
            \mathcal{O}(\eta^{n-2}), \quad k=2,
        \end{cases}
    \end{equation}
    where $\mathbf{v}=(x_{1},0,\cdots,0,x_{k},0,\cdots)$. Especially, $S_{n}(\mathbf{x})=S_{n}(\mathbf{v})$ for $n=1,2$. 
\end{prop}
\begin{proof}
    For Proposition \ref{prop-Ok1}, we can establish the following identity
    \begin{equation}
        \begin{split}
            \sum_{n=0}^{\infty}S_{n}(\mathbf{x})\epsilon^{n}&=\exp({\sum_{n=1}^{\infty}x_{n}\epsilon^{n}})\\
            &=\exp({\sum_{n=1}^{\infty}\eta^{-n}x_{n}(\epsilon\eta)^{n}})\\
            &=\sum_{n=0}^{\infty}S_{n}(x_{1}\eta^{-1},x_{2}\eta^{-2},\cdots)(\epsilon\eta)^{n}\\
            &=\sum_{n=0}^{\infty}\eta^{n}S_{n}(x_{1}\eta^{-1},x_{2}\eta^{-2},\cdots)\epsilon^{n},
        \end{split}
    \end{equation}
    and group terms according to the power of $\epsilon$. 
    To prove Proposition \ref{prop-OK2}, if $k\geq 3$, we proceed with the following calculations
    \begin{equation}\label{pro11}
        \begin{split}
            &\sum_{n=0}^{\infty}(S_{n}(x_{1}\eta^{-1},\cdots,x_{k-1}\eta^{-(k-1)},x_{k}\eta^{-k},x_{k+1}\eta^{-(k+1)},\cdots)-
            S_{n}(x_{1}\eta^{-1},\cdots,0,x_{k}\eta^{-k},0,\cdots))\epsilon^{n}\\
            =&\exp(x_{1}\eta^{-1}\epsilon+x_{k}\eta^{-k}\epsilon^{k})\left(
                \exp(\sum_{n=2,n\ne k}^{\infty}x_{n}(\frac{\epsilon}{\eta})^{n})-1
            \right)\\
            =&\exp(\mathcal{O}(1)\epsilon+\mathcal{O}(1)\epsilon^{k})\left(
                \exp(\sum_{n=2,n\ne k}^{\infty}\mathcal{O}(\eta^{-n+1})\epsilon^{n})-1
            \right)\\
            =&\left(\sum_{n=0}^{\infty}\mathcal{O}(1)\epsilon^{n}\right)
            \left(\mathcal{O}(\eta^{-1})\epsilon^{2}+\cdots\right)\\
            =&\sum_{n=2}^{\infty}\mathcal{O}(\eta^{-1})\epsilon^{n}.
        \end{split}
    \end{equation}
    If $k=2$, 
    \begin{equation*}
        \begin{split}
            &\sum_{n=0}^{\infty}(S_{n}(x_{1}\eta^{-1},x_{2}\eta^{-2},x_{3}\eta^{-3},\cdots)-
            S_{n}(x_{1}\eta^{-1},x_{2}\eta^{-2},0,\cdots))\epsilon^{n}\\
            =&\exp(x_{1}\eta^{-1}\epsilon+x_{2}\eta^{-2}\epsilon^{2})\left(
                \exp(\sum_{n=3}^{\infty}x_{n}(\frac{\epsilon}{\eta})^{n})-1
            \right)\\
            =&\exp(\mathcal{O}(1)\epsilon+\mathcal{O}(1)\epsilon^{2})\left(
                \exp(\sum_{n=3}^{\infty}\mathcal{O}(\eta^{-n+1})\epsilon^{n})-1
            \right)\\
            =&\left(\sum_{n=0}^{\infty}\mathcal{O}(1)\epsilon^{n}\right)
            \left(\mathcal{O}(\eta^{-2})\epsilon^{3}+\cdots\right)\\
            =&\sum_{n=3}^{\infty}\mathcal{O}(\eta^{-2})\epsilon^{n}.
        \end{split}
    \end{equation*}
    By utilizing Proposition \ref{prop-Ok1} and grouping the terms with respect to $\epsilon$, 
    we complete the proof.
\end{proof}

Based on Proposition \ref{prop-OK2}, it is established that the Schur polynomials can 
be expressed as simplified polynomials involving only two parameters $x_{1},x_{k}$ with error terms
when one of the parameters is large enough. This simplification provides the basis 
for our investigation into Okamoto polynomial hierarchies 
\cite{okamoto_studies_1986,kajiwara_determinant_1998, yang_rogue_2022}. To define the hierarchies, 
we consider a special form of Schur polynomials $p_j^{[m]}(z)$ which are defined by
\begin{equation}
    \sum_{j=0}^{\infty} p_j^{[m]}(z) \epsilon^j=\exp \left(z \epsilon+\epsilon^m\right),
\end{equation}
where $z\in\mathbb{C}$. 
Then we define $k$-type Okamoto polynomial hierarchies \cite{yang_rogue_2022} for $k=0,1$ respectively

\begin{equation}\label{Okamoto}
    W_{N}^{[k,m]}(z)=c_{N}^{[k]}\det(p_{3i-j-k}^{[m]})_{1\leq i,j \leq N}=c_{N}^{[k]}
    \left|
    \begin{matrix}
        p_{2-k}^{[m]}(z) & p_{1-k}^{[m]}(z) & \cdots & p_{3-N-k}^{[m]}(z)\\
        p_{5-k}^{[m]}(z) & p_{4-k}^{[m]}(z) & \cdots & p_{6-N-k}^{[m]}(z)\\
        \vdots & \vdots & \ddots & \vdots \\
        p_{3N-1-k}^{[m]}(z) & p_{3N-2-k}^{[m]}(z) & \cdots & p_{2N-k}^{[m]}(z)\\
    \end{matrix}
    \right|,
\end{equation}
where
\begin{equation}
    c_N^{[k]}=3^{-\frac{1}{2} N(N-1)} \frac{(2-k) ! (5-k) ! \cdots(3 N-1-k) !}{0 ! 1 ! \cdots(N-1) !}, 
\end{equation}
which ensure that the leading order terms of $W_{N}^{[k,m]}(z)$ with respect to $z$ are equal to $1$. 
To study the decomposition of rogue wave solutions, we also define 
\begin{equation}
    \begin{split}
    W_{N,1}^{[k,m]}(z)=c_{N}^{[k]}&\left(
    \left|
        \begin{matrix}
            p_{2-k}^{[m]}(z)  & \cdots & p_{4-N-k}^{[m]}(z) & p_{1-N-k}^{[m]}(z)\\
            p_{5-k}^{[m]}(z)  & \cdots & p_{7-N-k}^{[m]}(z)& p_{4-N-k}^{[m]}(z)\\
            \vdots & \vdots & \ddots & \vdots \\
            p_{3N-1-k}^{[m]}(z)  & \cdots & p_{2N+1-k}^{[m]}(z) & p_{2N-2-k}^{[m]}(z)\\
        \end{matrix}
    \right| \right.\\
    &\left.+
    \left|
        \begin{matrix}
            p_{2-k}^{[m]}(z)  & \cdots & p_{2-N-k}^{[m]}(z) & p_{3-N-k}^{[m]}(z)\\
            p_{5-k}^{[m]}(z)  & \cdots & p_{5-N-k}^{[m]}(z)& p_{6-N-k}^{[m]}(z)\\
            \vdots & \vdots & \ddots & \vdots \\
            p_{3N-1-k}^{[m]}(z)  & \cdots & p_{2N-1-k}^{[m]}(z) & p_{2N-k}^{[m]}(z)\\
        \end{matrix}
    \right|   
    \right).
    \end{split}
\end{equation}
The term $W_{N,1}^{[k,m]}(z)$ is the sum of two variations of $W_{N}^{[k,m]}(z)$. 
One variation is to subtract two from the indices of the elements in the penultimate column of $W_{N}^{[k,m]}(z)$, 
and the other variation is to subtract two from the indices of the elements in the last column of $W_{N}^{[k,m]}(z)$. 
For the CFL equations, the expression of rogue wave solutions in the outer region includes 
terms $W_{N,1}^{[k,m]}(z)$. 

Now we turn to the root structures of the Okamoto polynomial hierarchies \eqref{Okamoto}, 
which are important in understanding rogue wave patterns. 
Previous studies have demonstrated that all Okamoto polynomials (the case $m=2$) have simple 
roots \cite{kametaka_poles_1983,fukutani_special_2000}. 
When $m$ and $N$ are small, it can be observed that the nonzero roots of the Okamoto polynomial hierarchies are typically simple, 
while the zero roots may be multiple roots \cite{yang_rogue_2022}. 
However, whether all non-zero roots of the hierarchies are simple remains a conjecture. 
Nevertheless, some results have already been obtained. 
Yang et al. have studied the root distributions of the Okamoto polynomial hierarchies \cite{yang_rogue_2022}. 
This theorem reveals the symmetry in the patterns of rogue waves. 
To analyze the root distributions of Okamoto polynomial hierarchies, let $N_{0}$ be the remainder of $N$ divided by $m$, 
we define $(N_{1}^{[k]},N_{2}^{[k]})$ as follow: If $m\mod{3}\equiv 1$, 
\begin{equation}
    (N_{1}^{[0]},N_{2}^{[0]})=
    \begin{cases}
        (N_{0},0),                               & 0\leq N_{0} \leq [\frac{m}{3}],\\
        ([\frac{m}{3}],N_{0}-[\frac{m}{3}]), & [\frac{m}{3}]+1\leq N_{0} \leq 2 [\frac{m}{3}],\\
        (m-1-N_{0},m-1-N_{0}),                     & 2[\frac{m}{3}]+1\leq N_{0} \leq m-1,\\
    \end{cases}
\end{equation}
\begin{equation}
    (N_{1}^{[1]},N_{2}^{[1]})=
    \begin{cases}
        (0,N_{0}),                               & 0\leq N_{0} \leq [\frac{m}{3}],\\
        ([\frac{m}{3}]-1,N_{0}-1-[\frac{m}{3}]), & [\frac{m}{3}]+1\leq N_{0} \leq 2 [\frac{m}{3}]+1,\\
        (m-1-N_{0},m-N_{0}),                     & 2[\frac{m}{3}]+2\leq N_{0} \leq m-1.\\
    \end{cases}
\end{equation}
If $m\mod{3}\equiv 2$, 
\begin{equation}
    (N_{1}^{[0]},N_{2}^{[0]})=
    \begin{cases}
        (N_{0},0),                               & 0\leq N_{0} \leq [\frac{m}{3}],\\
        (N_{0}-1-[\frac{m}{3}],[\frac{m}{3}]), & [\frac{m}{3}]+1\leq N_{0} \leq 2 [\frac{m}{3}]+1,\\
        (m-1-N_{0},m-1-N_{0}),                     & 2[\frac{m}{3}]+2\leq N_{0} \leq m-1,\\
    \end{cases}
\end{equation}
\begin{equation}
    (N_{1}^{[1]},N_{2}^{[1]})=
    \begin{cases}
        (0,N_{0}),                               & 0\leq N_{0} \leq [\frac{m}{3}]+1,\\
        (N_{0}-1-[\frac{m}{3}],[\frac{m}{3}]+1), & [\frac{m}{3}]+2\leq N_{0} \leq 2 [\frac{m}{3}]+1,\\
        (m-1-N_{0},m-N_{0}),                     & 2[\frac{m}{3}]+2\leq N_{0} \leq m-1,\\
    \end{cases}
\end{equation}
where the symbol $[x]=\max_{n\in\mathbb{Z},n\leq x}{n}$. 
These notations are used to study the degree of zero roots of Okamoto polynomial hierarchies 
and analyze the rogue wave patterns in the inner region. 
The following theorem \cite{yang_rogue_2022} is hold:
\begin{thm}\label{Ok-root}
    Given an integer $m\geq2$, the Okamoto polynomial hierarchies $W_{N}^{[k,m]}(z)$ is monic with degree $N(N+1-k)$. 
    If $m$ is not a multiple of $3$, then $W_{N}^{[k,m]}(z)$ have the decomposition
    \begin{equation}\label{Ok-poly}
        W_{N}^{[k,m]}(z)=z^{N^{[k]}}q_{N}^{[k,m]}(z^{m}),
    \end{equation}
    where $q_{N}^{[k,m]}(\xi)$ is a monic polynomial with respect to $\xi$ with all real-value coefficients and a nonzero constant term. 
    The multiplicity of the zero root is 
    \begin{equation}
        N^{[k]}=N_{1}^{[k]}(N_{1}^{[k]}-N_{2}^{[k]}+1)+(N_{2}^{[k]})^{2}.
    \end{equation}
    If $m$ is a multiple of $3$, then 
    \begin{equation}
        W_{N}^{[k,m]}(z)=z^{N(N+1-k)}.
    \end{equation}
\end{thm}

Based on the Proposition \ref{prop-Ok1}, \ref{prop-OK2} and the root structures of Okamoto polynomial hierarchies in Theorem \ref{Ok-root}, 
we can analyze the rogue wave patterns in Section \ref{sec-pattern} if $m$ is not a multiple of $3$. The case of $m$ is a multiple of 
$3$ will be excluded from consideration, and we will see the reason in the proof of rogue wave patterns. 
When we consider the rogue wave decomposition in the inner region, the proof is similar 
to the Theorem \ref{Ok-root}. 

The root distributions indicate the positions in the rogue wave patterns. 
More specifically, with a linear transformation, the positions of rogue waves correspond to 
the root distributions of Okamoto polynomial hierarchies. The order of the rogue waves corresponds 
to the degree of roots. 

\subsection{Darboux transformation}
Now we turn to introducing the Darboux transformation \cite{ling_general_2018}, which is used to convert the Lax pair \eqref{CFL-lax} into a new one. 
For the new elements in the Lax pair \eqref{CFL-lax}, we denote it by adding superscript $^{[N]}$, such as the new potential functions $\mathbf{Q}^{[N]}$. 
By establishing a relationship between the original and new potential functions, together with initial seed solutions, 
we can construct a variety of new solutions. 
In Section \ref{sec-solution}, we will focus on specific parameter selections for the Darboux transformation, 
which enables us to generate rogue wave solutions by plane wave solutions.

We introduce the $N$-fold Darboux transformation $\mathbf{T}_N(\lambda;x,t)$, as presented in \cite{ling_general_2018}. 
Let $\mathbf{A}_i=|x_i\rangle \langle y_i|\mathbf{J}$, where $|x_i\rangle=(x_{i,1},x_{i,2},x_{i,3})^{T}$ and 
$|y_i\rangle=(y_{i,1},y_{i,2},y_{i,3})^{T}$ are three dimensional complex vectors, and 
$\mathbf{J}=\mathrm{diag} (1,-1,-\sigma)$. The transformation from  $|y_i\rangle$ to $|x_i\rangle$ is
\begin{equation*}
    \begin{split}
        \left[|x_{1,1}\rangle,|x_{2,1}\rangle,\cdots,|x_{N,1}\rangle\right]&  =
\left[|y_{1,1}\rangle,|y_{2,1}\rangle,\cdots,|y_{N,1}\rangle\right]\mathbf{B}^{-1},\,\, \mathbf{B}=(b_{ij})_{N\times N},\\
        \left[|x_{1,k}\rangle,|x_{2,k}\rangle,\cdots,|x_{N,k}\rangle\right]&  =
\left[|y_{1,k}\rangle,|y_{2,k}\rangle,\cdots,|y_{N,k}\rangle\right]\mathbf{M}^{-1},\,\,\mathbf{M}=(m_{ij})_{N\times N},\,\, k=2,3,
\end{split}
\end{equation*}
and the coefficients are given by
\begin{equation*}
b_{ij}=\frac{\langle y_i|\mathbf{J}|y_j\rangle}{\lambda_i^*-\lambda_j}+\frac{\langle y_i|\mathbf{J}\sigma_3|y_j\rangle}{\lambda_i^*+\lambda_j},\,\, m_{ij}=\frac{\langle y_i|\mathbf{J}|y_j\rangle}{\lambda_i^*-\lambda_j}-\frac{\langle y_i|\mathbf{J}\sigma_3|y_j\rangle}{\lambda_i^*+\lambda_j}. 
\end{equation*}
Let superscript $^{\dagger}$ denotes the complex conjugate and transposition, 
the $N$-fold Darboux transformation has the following form. 
\begin{thm}\label{DT}
By the following $N$-fold Darboux transformation 
\begin{equation}\label{n-fold-dt}
    \mathbf{T}_N(\lambda;x,t)=\mathbb{I}+\sum_{i=1}^{N}\left[\frac{\mathbf{A}_i}{\lambda-\lambda_i^*}-\frac{\sigma_3\mathbf{A}_i\sigma_3}{\lambda+\lambda_i^*}\right],
\end{equation}
the Lax pair \eqref{CFL-lax} can be converted into a new one. 
Then the B\"acklund transformation between old and new potential functions is
\begin{equation}\label{nfoldbt}
\mathbf{Q}^{[N]}=\mathbf{Q}+\sum_{i=1}^N(\mathbf{A}_i-\sigma_3\mathbf{A}_i\sigma_3),
\end{equation}
i.e. 
\begin{equation}
    u_{i}^{[N]}=u_{i}+2\mathbf{Y}_{i}\mathbf{M}^{-1}\mathbf{Y}^{\dagger},\quad i=1,2,
\end{equation}
where
\begin{equation*}
    \mathbf{Y}=\begin{pmatrix}
    y_{1,1} & y_{2,1} & \cdots & y_{N,1}
        \end{pmatrix},\,\,\mathbf{Y}_i=\begin{pmatrix}
        y_{1,i+1} & y_{2,i+1} & \cdots & y_{N,i+1} \\
        \end{pmatrix}.
\end{equation*}
\end{thm}
In this paper, we will consider $u_{i}\ne 0$, and it follows that
\begin{equation}\label{N-sol}
    u_{i}^{[N]}=u_{i}\frac{\det(\mathbf{M}+2u_{i}^{-1}\mathbf{Y}^{\dagger}\mathbf{Y}_{i})}{\det(\mathbf{M})},\quad i=1,2,
\end{equation}
since
\begin{equation*}
    1+\mathbf{Y}_{i}\mathbf{M}^{-1}\mathbf{Y}^{\dagger}=\frac{\det
    \begin{pmatrix}
        \mathbf{M}&\mathbf{Y}^{\dagger}\\
        \mathbf{Y}_{i}&1
    \end{pmatrix}
    }{\det(\mathbf{M})}.
\end{equation*}
We proceed to analyze the numerator and denominator of the obtained solution \eqref{N-sol}. 
Specifically, by selecting appropriate $|y_i\rangle$ and seed solutions $u_i$ in Theorem \ref{DT}, 
the elements of the numerator and denominator in \eqref{N-sol} exhibit quadric forms which are helpful in 
constructing the rogue wave solutions. 

\section{Rogue wave solutions}\label{sec-solution}
In Section \ref{sec-pre}, we have discussed the theorem regarding Okamoto polynomial hierarchies and the Darboux transformation for the CFL equations. 
In the subsequent section, we will utilize the properties of Schur polynomials 
and the root structures of Okamoto polynomial hierarchies to analyze the rogue wave patterns. 
Specifically, in this section, we employ Theorem \ref{DT} to construct the rogue wave solutions. 
To achieve this, we introduce the seed solutions $u_i$ and select specific $|y_i\rangle$ vectors.

\subsection{Seed solution and spectral characteristic polynomial}
We will consider the seed solutions in the form of plane wave solutions in Theorem \ref{DT}. 
Through these plane wave solutions, we can transform the Lax pair \eqref{CFL-lax} into a system with 
constant coefficients. By simultaneously diagonalizing the transformed matrices for $\mathbf{U}$ and $\mathbf{V}$ in \eqref{CFL-lax}, 
we can effectively solve the Lax pair \eqref{CFL-lax} and obtain the fundamental solutions. 
It is worth noting that the choices of the parameters $|y_i\rangle$ are connected to the fundamental solutions.

The fundamental solutions are determined by the spectral characteristic polynomial, 
which forms a three-sheet Riemann surface. We will investigate the properties at the branch points on the Riemann surface. 
These properties play a crucial role in determining the feasibility of constructing the rogue wave solutions 
through the Darboux transformation.

It is accessible to obtain the plane wave solutions for the CFL equations \eqref{CFL}:
\begin{equation}\label{seed}
    u_i^{[0]}=a_i{\rm e}^{{\rm i}\omega_i},\,\, i=1,2
\end{equation}
where
\begin{equation*}
\begin{split}
    \omega_1=&b_1x-\frac{1}{2}\left(2a_1^2+\sigma a_2^2-\frac{2}{b_1}+\sigma a_2^2\frac{b_2}{b_1}\right)t,  \\
    \omega_2=&b_2x-\frac{1}{2}\left(2\sigma a_2^2+a_1^2-\frac{2}{b_2}+a_1^2\frac{b_1}{b_2}\right)t,
\end{split}
\end{equation*}
the parameters $a_{i}$s are real numbers and $b_{i}$s are nonzero real numbers. 
Inserting the seed solutions \eqref{seed} into the Lax pair \eqref{CFL-lax}, 
introducing $z=1/\lambda^2$, we solve the Lax pair \eqref{CFL-lax} by ODE. Consider the parameter settings $a_{i}\ne 0$ and 
$b_{1}\ne b_{2}$, 
we have the fundamental solutions for the Lax pair \eqref{CFL-lax}
\begin{equation}\label{fund2}
    \mathbf{\Phi}(\lambda)=\mathbf{D}\mathbf{E}\mathrm{diag}\left({\rm e}^{\theta_1},{\rm e}^{\theta_2},{\rm e}^{\theta_3}\right),\,\, \mathbf{D}=\mathrm{diag}\left(1,{\rm e}^{{\rm i}\omega_1},{\rm e}^{{\rm i}\omega_2}\right),
\end{equation}
where
\begin{equation*}
\begin{split}
    \theta_i&={\rm i}(\kappa_i-z)\left(x+\frac{1}{2b_1b_2z}(\kappa_i-z+b_1+b_2)t\right),\,\, i=1,2,3\\
\end{split}
\end{equation*}
    and
\begin{equation*}
\begin{split}
    \mathbf{E}&=\begin{pmatrix}
            1 & 1 & 1 \\
            \frac{a_1b_1}{\lambda(\kappa_1+b_1)} & \frac{a_1b_1}{\lambda(\kappa_2+b_1)} & \frac{a_1b_1}{\lambda(\kappa_3+b_1)} \\
            \frac{a_2b_2}{\lambda(\kappa_1+b_2)} & \frac{a_2b_2}{\lambda(\kappa_2+b_2)} & \frac{a_2b_2}{\lambda(\kappa_3+b_2)} \\
        \end{pmatrix}.
\end{split}
\end{equation*}
The terms $\kappa_i, i=1,2,3$ satisfy the algebraic equation
\begin{equation}\label{chara-2}
    \frac{\kappa}{z}-2+\frac{a_1^2 b_1^2}{(\kappa+b_1)}+\frac{\sigma a_2^2b_2^2}{(\kappa+b_2)}=0.
\end{equation}
Note that $\lambda$ is the primary spectral parameter of Lax pair $\eqref{CFL-lax}$, but here we use the parameter $z$. 
The algebraic equation \eqref{chara-2} 
generate a three-sheet Riemann surface 
\begin{equation}\label{RS1}
    \mathcal{R}=\left\{(z,\kappa)\in \mathbb{S}^{2}:\frac{\kappa}{z}-2+\frac{a_1^2 b_1^2}{(\kappa+b_1)}+\frac{\sigma a_2^2b_2^2}{(\kappa+b_2)}=0\right\},
\end{equation}
with projection $p:(z,\kappa)\mapsto z$, where $\mathbb{S}$ is the Riemann sphere. 
Denote $\alpha=a_1^2b_1^2+\sigma a_2^2b_2^2$, $\beta=a_1^2b_1^2b_2+\sigma a_2^2b_1b_2^2$, $\gamma=b_1+b_2$, $\delta=b_1-b_2$, 
the branch points of \eqref{RS1} are determined by the following quartic equation with respect to $z$:
\begin{equation}\label{quantic}
    \begin{split}
        A_{4} {z}^{4}+ A_{3} z^{3}+A_{2} z^{2}+A_{1} z+A_{0}=0,
    \end{split}
\end{equation}
where the LHS is the discriminant of the spectral characteristic polynomial \eqref{chara-2} with respect to $\kappa$. The coefficients are
given by
\begin{equation*}
    \begin{split}
        A_{4}&=4 {\alpha}^{2}-16 \gamma \alpha+16 {\delta}^{2}+32 \beta ,\quad A_{3}=-4 {\alpha}^{3}+20 {\alpha}^{2}\gamma
    -20 \alpha {\delta}^{2}-12 \alpha {\gamma}^{2}+16 {\delta}^{2}
    \gamma-36 \alpha \beta+24 \beta \gamma, \\
    A_{2}&=3 {\alpha}^{2}{\delta}^{2}-2 {\alpha}^{2}{\gamma}^{2}-\alpha
    {\delta}^{2}\gamma-3 \alpha {\gamma}^{3}-2 {\delta}^{4}+6 {
    \delta}^{2}{\gamma}^{2}+18 \alpha \beta \gamma-18 \beta {\delta}^
    {2}+6 \beta {\gamma}^{2}-27 {\beta}^{2}, \\
    A_{1}&=-\frac{1}{4} {\gamma}^{4}\alpha+\frac{1}{2} {\gamma}^{3}\beta+{\gamma}^{3}{
        \delta}^{2}-\gamma {\delta}^{4}-\frac{3}{4} \alpha {\delta}^{4}+{\gamma}^{2
        }\alpha {\delta}^{2}-\frac{9}{2} \gamma \beta {\delta}^{2}, \quad
        A_{0}=\frac{1}{16} {\delta}^{6}+\frac{1}{16} {\delta}^{2}{\gamma}^{4}-\frac{1}{8} {\delta}^{4}{
            \gamma}^{2}.
    \end{split}
\end{equation*}
The complex roots (not real) of the quartic equation \eqref{quantic} corresponds to the rogue wave solutions of the 
CFL equations \eqref{CFL}. But generally, it is hard to analyze the roots of the algebraic equation \eqref{quantic}. 
The discriminant of \eqref{quantic} with respect to $z$ is useful to analyze the roots:
\begin{equation}\label{dis}
    \begin{split}
    \Delta\equiv\frac{1}{16}  \left(2 \beta+{\delta}^{2}-{\gamma}^{2}\right)\left[(\gamma \alpha-2 \beta)^2 -\alpha^2 \delta^2\right]\times \\
    \left[ \left( 324 {\delta}^{2}+108 {\gamma}^{2} \right) {\beta}^{2}+
\left( 54 {\alpha}^{2}{\delta}^{2}+18 {\alpha}^{2}{\gamma}^{2}-180
\alpha {\delta}^{2}\gamma-108 \alpha {\gamma}^{3}+288 {\delta}^{
4} \right) \beta\right.\\\left.-8 {\alpha}^{3}{\gamma}^{3}+27 {\delta}^{4}{\alpha}^
{2}-6 {\alpha}^{2}{\delta}^{2}{\gamma}^{2}+27 {\alpha}^{2}{\gamma}^{
4}-96 \alpha {\delta}^{4}\gamma+64 {\delta}^{6}
\right] ^{3}. 
    \end{split}
\end{equation}
The equation $\Delta=0$ can be solved with respect to $\beta$. 
Our main idea is to determine the cases of the roots of algebraic equation \eqref{quantic} by evaluating different values of $\beta$.
In this regard, we establish the following properties:
\begin{prop}\label{root-classify}
    If the parameters $(a_1,a_2,b_1,b_2)$ belong to $\Omega=\{(a_1,a_2,b_1,b_2)| a_1,a_2\neq 0,\,\, \sigma b_2\neq 2a_2^{-2}-b_1a_1^{2}a_2^{-2},\, b_1\neq b_2\}$, 
    there are several cases for the roots of \eqref{quantic}:
    \begin{enumerate}
        \item \label{ca1}If $3 {\alpha}^{2}-4\gamma \alpha-4 {\delta}^{2}\geq0$, let $\beta_1^{[a]}\leq \beta_2^{[a]}\leq\beta_3^{[a]}\leq\beta_4^{[a]}\leq\beta_5^{[a]}$ 
        are the real roots of $\Delta$, there are three cases:
        \begin{enumerate}
            \item \label{ca1a}If $\beta\in (-\infty,\beta_1^{[a]})\cup(\beta_2^{[a]},\beta_3^{[a]})\cup(\beta_4^{[a]},\beta_5^{[a]})$, 
            we obtain two real roots and a pair of complex conjugate roots (Fig. \ref{rootp123}-a).
            \item \label{ca1b}If $\beta\in (\beta_1^{[a]},\beta_2^{[a]})\cup(\beta_3^{[a]},\beta_4^{[a]})\cup(\beta_5^{[a]},\infty)$, 
            we obtain four real roots (Fig. \ref{rootp123}-b).
            \item \label{ca1c}If $\beta\in\{\beta_1^{[a]},\beta_2^{[a]},\beta_3^{[a]},\beta_4^{[a]},\beta_5^{[a]}\}$, we obtain
            one, two, or three real roots, or a pair of complex conjugate roots, 
            or one real root and a pair of complex conjugate roots (Fig. \ref{rootp123}-c).
        \end{enumerate}
        \item \label{ca2}If $3 {\alpha}^{2}-4\gamma \alpha-4 {\delta}^{2}<0$, let $\beta_1^{[a]}\leq \beta_2^{[a]}\leq\beta_3^{[a]}$ 
        are the real roots of $\Delta$, there are three cases:
        \begin{enumerate}
            \item \label{ca2a}If $\beta\in (-\infty,\beta_1^{[a]})\cup(\beta_2^{[a]},\beta_3^{[a]})$, we can 
            obtain two real roots and a pair of complex conjugate roots (Fig. \ref{rootn123}-a).
            \item \label{ca2b}If $\beta\in (\beta_1^{[a]},\beta_2^{[a]})\cup (\beta_3^{[a]},\infty)$, we can 
            obtain four real roots or two pairs of complex conjugate roots (Fig. \ref{rootn123}-b, c).
            \item \label{ca2c}If $\beta\in\{\beta_1^{[a]},\beta_2^{[a]},\beta_3^{[a]}\}$, we obtain
            one, two, or three real roots, or a pair of complex conjugate roots, 
            or one real root and a pair of complex conjugate roots (Fig. \ref{rootn123}-d).
        \end{enumerate}
    \end{enumerate}
    \end{prop}
    \begin{proof}
        We use the discriminant \eqref{dis} to study the quartic equation \eqref{quantic}, which is a real quartic equation about $z$. 
    \begin{enumerate}
        \item If $3 {\alpha}^{2}-4\gamma \alpha-4 {\delta}^{2}\geq0$, we can obtain the following roots of $\Delta=0$:
        \begin{equation*}
            \begin{split}
                \beta_1^{[b]}= &\frac{1}{2}({\gamma}^{2}-{\delta}^{2}),  \,\,
                \beta_2^{[b]}= \frac{1}{2}\left( \gamma-\delta \right) \alpha,\,\,
                \beta_3^{[b]}= \frac{1}{2}\left( \gamma+\delta \right) \alpha,\\
                \beta_4^{[b]}=&\frac { \left( -9 {\delta}^{2}-3 {\gamma}^{2} \right) {\alpha
        }^{2}+ \left( 30 {\delta}^{2}\gamma+18 {\gamma}^{3} \right) \alpha-
        48 {\delta}^{4}+ \left|  \left( 3 \alpha-8 \gamma \right) {\delta}^
        {2}+\alpha {\gamma}^{2} \right| \sqrt {3}\sqrt {3 {\alpha}^{2}-4
        \gamma \alpha-4 {\delta}^{2}}}{36(3 {\delta}^{2}+{\gamma}^{2})},
        \\
        \beta_5^{[b]}=&\frac { \left( -9 {\delta}^{2}-3 {\gamma}^{2} \right) {\alpha
        }^{2}+ \left( 30 {\delta}^{2}\gamma+18 {\gamma}^{3} \right) \alpha-
        48 {\delta}^{4}- \left|  \left( 3 \alpha-8 \gamma \right) {\delta}^
        {2}+\alpha {\gamma}^{2} \right| \sqrt {3}\sqrt {3 {\alpha}^{2}-4
        \gamma \alpha-4 {\delta}^{2}}}{36(3 {\delta}^{2}+{\gamma}^{2})}.
            \end{split}
        \end{equation*}
        We rearrange the roots $\beta_1^{[a]}\leq \beta_2^{[a]}\leq\beta_3^{[a]}\leq\beta_4^{[a]}\leq\beta_5^{[a]}$, then
        \begin{enumerate}
            \item If $\beta\in (-\infty,\beta_1^{[a]})\cup(\beta_2^{[a]},\beta_3^{[a]})\cup(\beta_4^{[a]},\beta_5^{[a]})$, then $\Delta<0$. 
            The quartic equation \eqref{quantic} has two real roots and a pair of complex conjugate roots.
            \item If $\beta\in (\beta_1^{[a]},\beta_2^{[a]})\cup(\beta_3^{[a]},\beta_4^{[a]})\cup(\beta_5^{[a]},\infty)$, then 
            $\Delta>0$. The quartic equation \eqref{quantic} could have four real roots or
            two pairs of complex conjugate roots. Since the equation \eqref{quantic} is a square equation 
            with respect to the parameter $\beta$, solving the equation, the roots are
            \begin{equation}\label{beta12}
                \begin{split}
                \beta_1(z)=&\frac{64 z^{3}+ \left( 48 \gamma-72 \alpha \right) z^{2}+ \left( 36 \gamma \alpha-36 {\delta}^{2}+12 {\gamma}^{2}
            \right) z-9 {\delta}^{2}\gamma+{\gamma}^{3}+\sqrt { \Delta_{1} ^{3}}}{108 z}
            , \\
            \beta_2(z)=&\frac{64 z^{3}+ \left( 48 \gamma-72 \alpha \right) z^{2}+ \left( 36 \gamma \alpha-36 {\delta}^{2}+12 {\gamma}^{2}
            \right) z-9 {\delta}^{2}\gamma+{\gamma}^{3}-\sqrt { \Delta_{1} ^{3}}}{108 z},
                \end{split}
            \end{equation}
            where $\Delta_{1}\equiv-12
            \alpha z+3 {\delta}^{2}+{\gamma}^{2}+8 \gamma z+16
            {z}^{2}$. 
            To distinguish these two cases, we need to study the function $\beta_1(z)$ and $\beta_2(z)$. If the union 
            of the range for $\beta_1(z)$ and $\beta_2(z)$ is $\mathbb{R}$, 
        then there exists a real root of equation \eqref{quantic}. There would not exist two pairs of complex conjugate roots. 
        Denote the roots of the equation $\Delta_{1}=0$ are
        \begin{equation}
            \begin{split}
                z_{1}=\frac{3\alpha-2\gamma+\sqrt{3}\sqrt{3 {\alpha}^{2}-4\gamma \alpha-4 {\delta}^{2}}}{8},\\
                z_{2}=\frac{3\alpha-2\gamma-\sqrt{3}\sqrt{3 {\alpha}^{2}-4\gamma \alpha-4 {\delta}^{2}}}{8},
            \end{split}
        \end{equation}
        we get $z_{1}z_{2}>0$ and
        $\beta_{1}(z_1)=\beta_{2}(z_{1}),\beta_{1}(z_2)=\beta_{2}(z_{2})$. By direct calculation, we obtain the limit
        $\lim_{z\to+\infty}\beta_{1}(z)=\lim_{z\to-\infty}\beta_{2}(z)=+\infty$ and 
        \begin{equation}
            \begin{split}
                \lim_{z\to 0}z\beta_{1}(z)=\frac{-\gamma(9\delta^2-\gamma^2)+|3\delta^2+\gamma^{2}|\sqrt{3\delta^2+\gamma^{2}}}{108},\\
                \lim_{z\to 0}z\beta_{2}(z)=\frac{-\gamma(9\delta^2-\gamma^2)-|3\delta^2+\gamma^{2}|\sqrt{3\delta^2+\gamma^{2}}}{108}.
            \end{split}
        \end{equation}
        Thus $\lim_{z\to 0^{+}}z\beta_{1}(z)>0$ and $\lim_{z\to 0^{+}}z\beta_{2}(z)<0$. If $z_{1}>0$, since $\beta_{1}(z)$ and $\beta_{2}(z)$ 
        are continuous on $(0,z_{2}]$, the range contains $[\beta_{1}(z_{2}),+\infty)$ for $\beta_{1}(z)$ and $(-\infty,\beta_{2}(z_{2})]$ 
        for $\beta_{2}(z)$. If $z_{1}<0$, the range contains $(-\infty,\beta_{1}(z_{1})]$ for $\beta_{1}(z)$ 
        and $[\beta_{2}(z_{1}),+\infty)$ for $\beta_{2}(z)$. 
        Hence for any $\beta$, there exists a real $z$ such that the quantic equation \eqref{quantic} is valid. Therefore,
        there are no two pairs of complex conjugate roots for \eqref{quantic}.
        \item If $\beta\in\{\beta_1^{[a]},\beta_2^{[a]},\beta_3^{[a]},\beta_4^{[a]},\beta_5^{[a]}\}$, then $\Delta=0$. We obtain
        one, two, or three real roots, or a pair of complex conjugate roots, 
        or one real root and a pair of complex conjugate roots.
        \end{enumerate}
        \item If $3 {\alpha}^{2}-4
        \gamma \alpha-4 {\delta}^{2}<0$, we merely obtain the real roots $\beta_1^{[b]}$, 
        $\beta_2^{[b]}$ and $\beta_3^{[b]}$ for the equation $\Delta=0.$ 
        Rearranging the roots $\beta_1^{[a]}\leq \beta_2^{[a]}\leq\beta_3^{[a]}$, then
        \begin{enumerate}
            \item If $\beta\in (-\infty,\beta_1^{[a]})\cup(\beta_2^{[a]},\beta_3^{[a]})$, then $\Delta<0$. There are 
            two real roots and a pair of complex conjugate roots. 
            \item If $\beta\in (\beta_1^{[a]},\beta_2^{[a]})\cup (\beta_3^{[a]},\infty)$, then $\Delta>0$. 
            In this case, we can not distinguish whether there exists a real $z$ 
            such that the quartic equation \eqref{quantic} is valid just using these roots. There are four real roots or two pairs of complex conjugate roots.
            \item If $\beta\in\{\beta_1^{[a]},\beta_2^{[a]},\beta_3^{[a]}\}$, then $\Delta=0$. There are
            one, two, or three real roots, or a pair of complex conjugate roots, 
            or one real root and a pair of complex conjugate roots.
        \end{enumerate}
    \end{enumerate}
\end{proof}

Below we provide some examples. In Figure \ref{rootp123}, 
the three subfigures Figure (\ref{rootp123}-a), Figure (\ref{rootp123}-b) and Figure (\ref{rootp123}-c) correspond to
the condition (\ref{ca1a}), (\ref{ca1b}) and (\ref{ca1c}) 
in Proposition \ref{root-classify} respectively. 
In Figure \ref{rootn123}, the subfigure Figure(\ref{rootn123}-a) corresponds to the condition (\ref{ca2a}), 
the subfigures Figure(\ref{rootn123}-b) and Figure (\ref{rootn123}-c) correspond to the condition 
(\ref{ca2b}) and the subfigure Figure (\ref{rootn123}-c) corresponds to the condition (\ref{ca2c}) in Proposition \ref{root-classify}. 
\begin{figure}[htbp]
    \centering
    \includegraphics[scale=0.18]{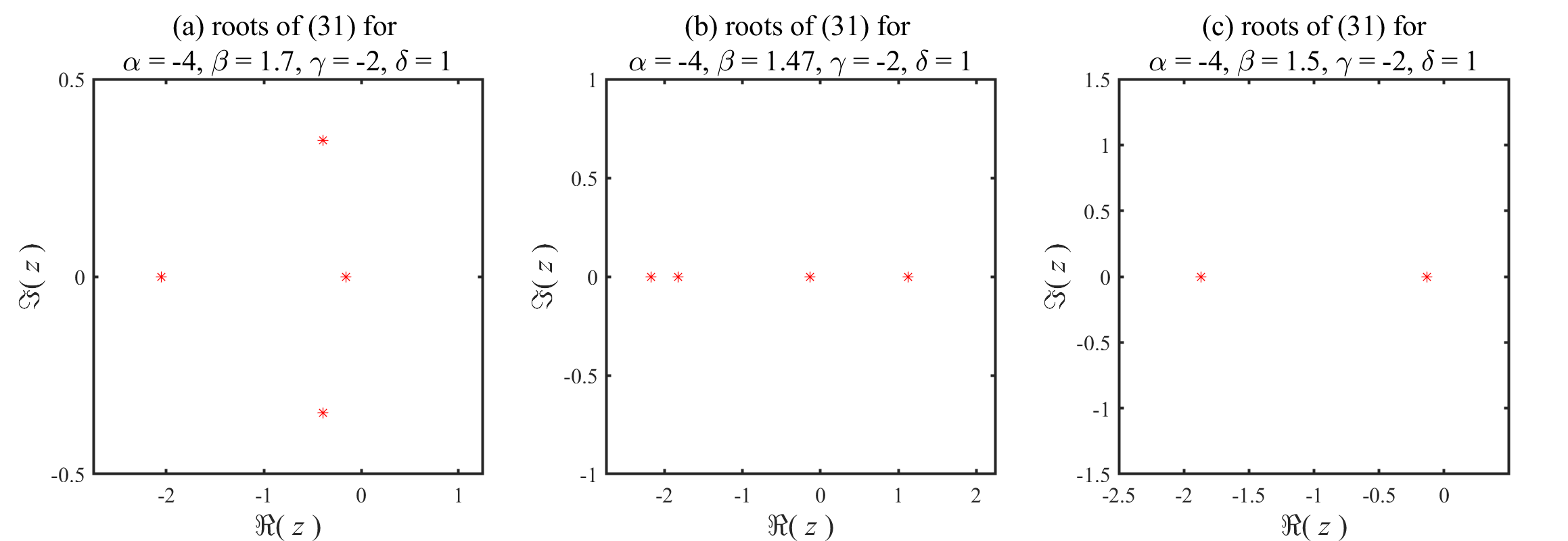}
    \caption{The root distributions of Equation \eqref{quantic}. Parameter settings: (a) $(\alpha,\beta,\gamma,\delta)=(-4,1.7,-2,1)$,
    (b) $(\alpha,\beta,\gamma,\delta)=(-4,1.47,-2,1)$ and (c) $(\alpha,\beta,\gamma,\delta)=(-4,1.5,-2,1)$.}
    \label{rootp123}
\end{figure}

\begin{figure}[htbp]
    \centering
    \includegraphics[scale=0.17]{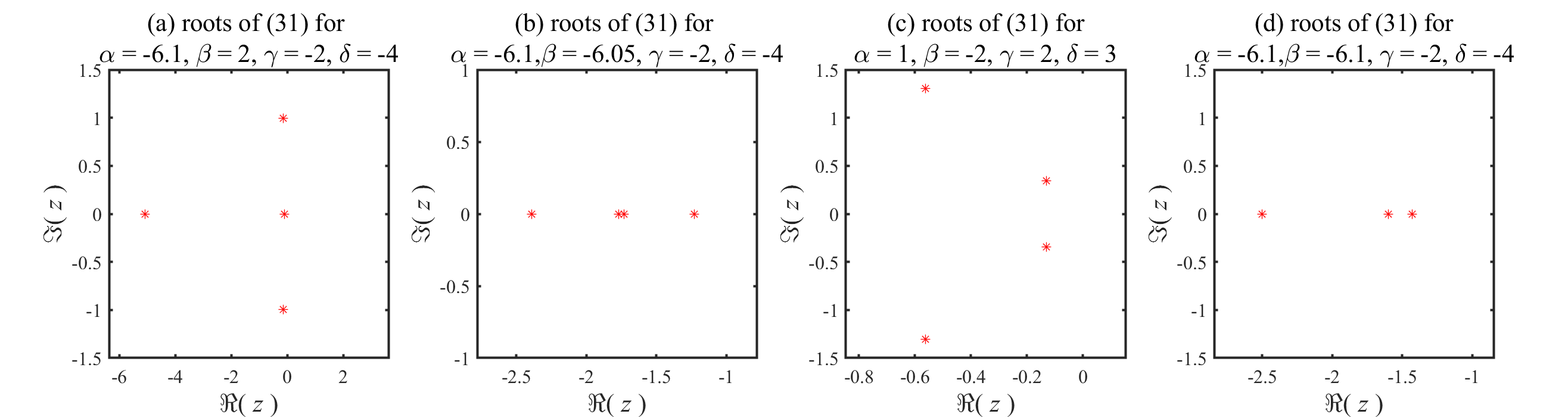}
    \caption{The root distributions of Equation \eqref{quantic}. Parameter settings: (a) $(\alpha,\beta,\gamma,\delta)=(-6.1,-6.05,-2,-4)$,
    (b) $(\alpha,\beta,\gamma,\delta)=(-6.1,2,-2,-4)$, (c) $(\alpha,\beta,\gamma,\delta)=(1,-2,2,3)$ and 
    (d) $(\alpha,\beta,\gamma,\delta)=(-6.1,-6.1,-2,-4)$. }
    \label{rootn123}
\end{figure}

Now we concentrate on the branch point $(z^{[0]},\kappa^{[0]})\in \mathcal{R}$ \eqref{RS1} of multiplicity three. 
By straightforward calculation, we obtain 
\begin{equation}\label{mult3}
    \begin{split}
        \kappa^{[0]}&={\frac {b_{{1}}b_{{2}} \left( -b_{{1}}-b_{{2}}+{\rm i}\sqrt {3} \left(
        b_{{1}}-b_{{2}} \right)  \right) }{b_{1}^{2}-b_{{1}}b_{{2}}+b_{2
        }^{2}}},
        \\        
        z^{[0]}&=\frac{b_{1}}{2}-\frac{b_{2}}{4}-\frac{3}{4}{\frac{b_{2}^{2}\left(2b_{1}
        -b_{2}\right)}{b_{1}^{2}-b_{1}b_{2}+b_{2}^{2}}}+\frac{3}{4}{\frac {{\rm i}b_{1} b_{2}\left( b_{1}-b_{2} \right)\sqrt {3}}{b_{1}^{2}-b_{1}b_{2}+b_{2}^{2}}},\,\,\,\, b_1,b_2>0.
    \end{split}
\end{equation}
As we can see, the imaginary parts of $z^{[0]}$ and $\kappa^{[0]}$ are not zero. 
Moreover, $a_{i}$s can be expressed by $b_{i}$s:
\begin{equation*}
\begin{split}
    a_1=\frac{\sqrt {2 b_{{1}}} \left( b_{{1}}-b_{{2}} \right) }{b_{
{1}}^{2}-b_{{1}}b_{{2}}+b_{{2}}^{2}},
\quad
    a_2={\frac{\sqrt {2 b_{{2}}} \left( b_{{1}}-b_{{2}} \right) }{b_{
{1}}^{2}-b_{{1}}b_{{2}}+b_{{2}}^{2}}}.
\end{split}
\end{equation*}
Hence the branch point and $a_{i}$s can be represented by $b_{i}$s. 
It is routine to verify that \eqref{mult3} satisfies the equation \eqref{quantic}. 
Since $z^{[0]}$ is a complex root (not real), we can generate high-order rogue waves at the point $(z^{[0]},\kappa^{[0]})$. 

To construct the local coordinate chart, we expand $(z,\kappa)\in\mathcal{R}$ at $(z^{[0]},\kappa^{[0]})$ 
with the following form with respect to $\epsilon$:
\begin{equation}\label{coord}
    \begin{split}
        z=z(\epsilon)=z^{[0]}+z^{[1]}\epsilon^3, \quad
        \kappa=\kappa(\epsilon)=\kappa^{[0]}-2z^{[1]}\epsilon \mu(\epsilon),
    \end{split}
\end{equation}
where
\begin{equation*}
    \begin{split}
        z^{[1]}=-\frac{1}{2}{\frac {b_{{1}}b_{{2}} \left( b_{{1}}-b_{{2}} \right) }{b_{{1}}^{2}-
b_{{1}}b_{{2}}+b_{{2}}^{2}}}, \quad
\mu(\epsilon)=\sum_{i=1}^{\infty}\mu_i\epsilon^{i-1}.
    \end{split}
\end{equation*}
Substituting $(z,\kappa)$ \eqref{coord} into the spectral characteristic polynomial \eqref{chara-2}, it leads to the recursive relation 
about $\kappa$: 
\begin{equation}\label{expand-f}
    \mu^3+\epsilon^2\mu^2+{\rm i}\sqrt{3}\epsilon \mu-1=0.
\end{equation}
Hence the coefficients $\mu_i$ can be determined through the following recursive relation:
\begin{equation*}
    \mu_1=1,\,\, \mu_{i}=-\frac{1}{3}\left(\sum_{j+k+l=i,\,0\leq j,k,l\leq i-1}\mu_j\mu_k\mu_l+\sum_{j+k=i-2,\,\,j,k\geq 0}\mu_j\mu_k+{\rm i}\sqrt{3}\mu_{i-1}\right),\,\, i\geq 2 .
\end{equation*}
The first several coefficients are
\begin{equation*}
    \mu_1=1,\,\, \mu_2=-\frac{\sqrt{3}}{3}{\rm i},\,\, \mu_3=-\frac{1}{3},\,\,
    \mu_4=\frac{2}{27}\sqrt{3}{\rm i},\,\, \mu_5=\frac{1}{27},\,\, \mu_6=0,\,\, \mu_7=\frac{1}{3^5},\,\, \mu_8=-\frac{\sqrt{3}{\rm i}}{3^6},\,\, \mu_9=0.
\end{equation*}
It can be verified that the convergence domain of the series $\mu(\epsilon)$ is $|\epsilon|<\sqrt{3}$. With the local coordinate 
chart \eqref{coord} at $(z^{[0]},\kappa^{[0]})$, the roots of \eqref{chara-2} are 
$\kappa_{i}=\kappa(\epsilon\omega^{i-1}),i=1,2,3$, where $\omega={\rm e}^{2\pi{\rm i}/3}$ is a root of equation $\omega^{3}=1$. 

Now we turn to constructing the rogue wave solutions. 
Note that the Darboux transformation can be used to generate soliton solutions \cite{ling_general_2018} for a non-branch point 
on the Riemann surface $\mathcal{R}$. 
The rogue wave solutions are generated at the branch point of multiplicity two and three, 
as stated in \cite{ye_general_2019}. 
In this paper, 
we only focus on the case multiplicity of three. We will investigate the case of multiplicity two in future work. 

\subsection{The determinant representation of rogue wave solutions}
Based on the seed solutions \eqref{seed}, we consider $|y_{s}\rangle=\mathbf{\Phi}(\lambda_{s})(c_{s,1},c_{s,2},c_{s,3})^{T},s=1,2,\cdots,N$ 
in Theorem \ref{DT}. For the spectral parameters $\lambda=\lambda_{s}$ in \eqref{fund2}, we denote $\kappa_{l}=\kappa_{l}^{(s)}$ and 
$\theta_{l}=\theta_{l}^{(s)}$. 
If the parameters $(a_1,a_2,b_1,b_2)$ belong to $\Omega=\{(a_1,a_2,b_1,b_2)| a_1,a_2\neq 0,\,\, \sigma b_2\neq 2a_2^{-2}-b_1a_1^{2}a_2^{-2},\, b_1\neq b_2\}$, 
then the determinant elements of the numerator and denominator defined in \eqref{N-sol} 
have the following quadric forms which are given in \cite{ling_general_2018}: 
    \begin{equation}
        \begin{split}
            m_{rs}&=\begin{pmatrix}
                c_{r,1}&c_{r,2}&c_{r,3}
            \end{pmatrix}^{*}
            (z^{[0]}_{k,l})_{1\leq k,l \leq 3}
            \begin{pmatrix}
                c_{s,1}\\c_{s,2}\\c_{s,3}
            \end{pmatrix},\\
            m_{rs}+2a_{i}^{-1}{\rm e}^{-{\rm i}\omega_{i}}y_{r,1}^{*}y_{s,i+1}&=\begin{pmatrix}
                c_{r,1}&c_{r,2}&c_{r,3}
            \end{pmatrix}^{*}
            (z^{[1]}_{k,l})_{1\leq k,l \leq 3}
            \begin{pmatrix}
                c_{s,1}\\c_{s,2}\\c_{s,3}
            \end{pmatrix},
        \end{split}
    \end{equation}
    where
    \begin{equation}
        \begin{split}\label{generating-function}
            z^{(0)}_{k,l}&=\frac{2}{\lambda_{s}}\frac{\kappa_{k}^{(r)*}}{\kappa_{l}^{(s)}-\kappa_{k}^{(r)*}}{\rm{e}}^{\theta_{l}^{(s)}+\theta_{k}^{(r)*}},\\
            z^{(1)}_{k,l}&=\frac{2}{\lambda_{s}}\frac{\kappa_{l}^{(s)}}{\kappa_{l}^{(s)}-\kappa_{k}^{(r)*}}\frac{\kappa_{k}^{(r)*}+b_{i}}{\kappa_{l}^{(s)}+b_{i}}{\rm{e}}^{\theta_{l}^{(s)}+\theta_{k}^{(r)*}}.
        \end{split}
    \end{equation}

Since dividing both the numerator and denominator by the same factor does not change the value of the solutions \eqref{N-sol}, 
we will consider two new elements by discarding the factors $2/\lambda_{s}$:
\begin{equation}\label{quadric-form}
    \begin{split}
        m_{rs}^{(0)}&=\frac{\lambda_{s}}{2}m_{rs},\\
        m_{rs}^{(1)}&=\frac{\lambda_{s}}{2}(m_{rs}+2a_{i}^{-1}{\rm e}^{-{\rm i}\omega_{i}}y_{r,1}^{*}y_{s,i+1}).
    \end{split}
\end{equation}
Using the above formulas \eqref{quadric-form}, we can analyze the concrete form of the solution \eqref{N-sol}. 
To obtain high-order rogue wave solutions, 
we use local coordinate chart \eqref{coord} at $(z^{[0]},\kappa^{[0]})$, then the terms 
$\kappa_{l}^{(s)}=\kappa(\epsilon_{s}\omega^{l-1})$. Since $\theta_{l}^{(s)}$ are the functions of $\kappa_{l}^{(s)}$ 
and $\lambda_{s}$, we also have $\theta_{l}^{(s)}=\theta(\epsilon_{s}\omega^{l-1})$.

On the other hand, we need to set special $(c_{s,1},c_{s,2},c_{s,3})$ in $|y_{s}\rangle$ to construct the rogue wave solutions. The idea arises from the calculation 
of the limit. 
If we take 
\begin{equation}\label{ys0-discussion}
    |y_{s}\rangle=|y_{s}^{(0)}\rangle:=\mathbf{\Phi}(\lambda_{s})(c^{(0)}_{s}(\epsilon_{s}),\omega c^{(0)}_{s}(\omega\epsilon_{s}),\omega^{2}c^{(0)}_{s}(\omega^{2}\epsilon_{s}))^{T} ,
\end{equation}
where $\mathbf{\Phi}(\lambda_{s})$ is defined in \eqref{fund2}. Then we can set 
\begin{equation*}
    \mathbf{\Phi}(\lambda_{s})=(\Phi_{s}(\epsilon_{s}),\Phi_{s}(\omega\epsilon_{s}),\Phi_{s}(\omega^{2}\epsilon_{s})), 
\end{equation*}
where $\Phi_{s}(\epsilon_{s})$ is a column vector. Expanding $\Phi_{s}(\epsilon_{s})=\sum_{i=0}^{\infty}\Phi_{s}^{[i]}\epsilon_{s}^{i}$, 
using $1+\omega+\omega^{2}=0$, we obtain 
\begin{equation}
    |y_{s}\rangle=3\sum_{k=1}^{\infty}\Phi_{s}^{[3k-1]}\epsilon_{s}^{3k-1}\mathrm{diag}\left(c^{(0)}_{s}(\epsilon_{s}),c^{(0)}_{s}(\omega\epsilon_{s}),c^{(0)}_{s}(\omega^{2}\epsilon_{s})\right) , 
\end{equation}
which has only $(3k-1)$th order coefficients with respect to $\epsilon_{s}$. 
Furthermore, dividing both the numerator and denominator by the same coefficients in \eqref{N-sol} does not alter the value 
of the solutions \eqref{N-sol}. 
Considering $|y_{s}\rangle/\epsilon_{s}^{2}$, we can obtain a solution that only takes into account the $(3k-1)$th order terms of $\epsilon_{s}$. 
Similarly, if we consider 
\begin{equation}\label{ys1-discussion}
    |y_{s}\rangle=|y_{s}^{(1)}\rangle:=\mathbf{\Phi}(\lambda_{s})(c^{(1)}_{s}(\epsilon_{s}),\omega^{2} c^{(1)}_{s}(\omega\epsilon_{s}),\omega c^{(1)}_{s}(\omega^{2}\epsilon_{s}))^{T}
\end{equation}
and $|y_s\rangle/\epsilon_{s}$, then the solution is only in terms of the $(3k-2)$th order coefficients about $\epsilon_{s}$.
Next, we will conduct precise calculations.

Now we need to introduce additional internal parameters, by considering another form of $c^{(l)}_{s}(\epsilon_{s}), l=0,1$. To simplify 
the notations, we consider a function $c_{s}(\epsilon_{s})$ firstly. 
Let $\chi^{[i]}\in \mathbb{C}$, 
we consider $c_{s}(\epsilon_{s})={\rm e}^{\sum_{i=1}^{\infty}\chi^{[i]}\epsilon_{s}^{i}}$ and $\chi^{[3i]}=0, i\geq 1$. 
Under the local coordinate chart \eqref{coord}, 
define $\vartheta(\epsilon_{s})=\theta(\epsilon_{s})+\sum_{i=1}^{\infty}\chi^{[i]}\epsilon_{s}^{i}$, we obtain 
\begin{equation}\label{exp-theta}
\vartheta(\epsilon_{s})={\rm i}\kappa(\epsilon_{s})\left[x+\left(\frac{\kappa(\epsilon_{s})+b_1+b_2}{z(\epsilon_{s})}-2\right)\frac{t}{2b_1b_2}\right]+\ln(c_{s}(\epsilon_{s}))
=\sum_{i=0}^{\infty}\vartheta^{[i]}\epsilon_{s}^{i}
\end{equation}
where
\begin{equation*}
\begin{split}
    \vartheta^{[0]}=&{\rm i}\kappa^{[0]}\left[x+\left(\frac{c_1}{z^{[0]}}-2\right)\frac{t}{2b_1b_2}\right],  \\
    \vartheta^{[1]}=&-2{\rm i}z^{[1]}\mu_1\left[x+\left(\frac{\kappa^{[0]}+c_1}{z^{[0]}}-2\right)\frac{t}{2b_1b_2}\right]+\chi^{[1]}, \\
\vartheta^{[i]}=&{\rm i}\left\{\left[-2z^{[1]}\sum_{k+3l=i,k\geq1,l\geq0}\mu_kc_2^{l}+c_1
c_2^{i/3}\delta_{i\mod3,0}\right]\frac{\kappa^{[0]}t}{2b_1b_2z^{[0]}}-2z^{[1]}\mu_i\left[x+\left(\frac{c_1}{z^{[0]}}-2\right)\frac{t}{2b_1b_2}\right]\right. \\
&\left.-\frac{z^{[1]}t}{b_1b_2z^{[0]}}\sum_{m+n=i}\mu_m\left[-2z^{[1]}\sum_{k+3l=n,k\geq0,l\geq0}\mu_k
c_2^{l}+c_1c_2^{n/3}\delta_{n\mod3,0}\right]\right\}+\chi^{[i]},\,\,\, i\geq 2, 
\end{split}
\end{equation*}
the terms $c_1=\kappa^{[0]}+b_1+b_2$, $c_2=-z^{[1]}/z^{[0]}$ and $\delta_{i\mod 3,0}$ is the Kronecker's delta. 

With the above preliminaries, we can construct the high-order rogue waves at the branch point of multiplicity three. 
Taking $|y_{s}\rangle=|y_{s}^{(0)}\rangle$ and $c^{(0)}_{s}(\epsilon_{s})=c_{s}(\epsilon_{s})$ that defined in \eqref{ys0-discussion}
and \eqref{exp-theta} for all $s=1,2,\cdots N$, since $\kappa_{l}^{(s)}=\kappa(\epsilon_{s}\omega^{l-1})$, 
we consider the following two functions naturally with respect to $\epsilon_{s}$ and $\epsilon_{r}^{*}$:
\begin{equation}\label{denom-expand}
    \begin{split}
        \mathcal{M}(\epsilon_{s},\epsilon_{r}^{*})&=\frac{\hat{\kappa}(\epsilon_{r}^{*}){\rm e}^{\vartheta(\epsilon_{s})+\hat{\vartheta}(\epsilon_{r}^{*})}}{\kappa(\epsilon_{s})-\hat{\kappa}(\epsilon_{r}^{*})}=\frac{\kappa^{[0]*}}{\kappa^{[0]}-\kappa^{[0]*}}{\rm e}^{\vartheta^{[0]}+\vartheta^{[0]*}}\sum_{k=0,l=0}^{\infty,\infty}M_{k,l}\epsilon_{r}^{*k}\epsilon_{s}^l,  \\
        \mathcal{G}(\epsilon_{s},\epsilon_{r}^{*})&=\frac{\hat{\kappa}(\epsilon_{r}^{*})+b_i}{\kappa(\epsilon_{s})+b_i}\frac{\kappa(\epsilon_{s}){\rm e}^{\vartheta(\epsilon_{s})+\hat{\vartheta}(\epsilon_{r}^{*})}}{\kappa(\epsilon_{s})-\hat{\kappa}(\epsilon_{r}^{*})}=\frac{\kappa^{[0]*}+b_i}{\kappa^{[0]}+b_i}\frac{\kappa^{[0]}}{\kappa^{[0]}-\kappa^{[0]*}}{\rm e}^{\vartheta^{[0]}+\vartheta^{[0]*}}\sum_{k=0,l=0}^{\infty,\infty}G_{k,l}\epsilon_{r}^{*k}\epsilon_{s}^l,
    \end{split}
\end{equation}
where $\hat{\kappa}(\epsilon_{r}^{*})=\kappa(\epsilon_{r})^{*}$ and $\hat{\vartheta}(\epsilon_{r}^{*})=\vartheta(\epsilon_{r})^{*}$. 
Then we calculate the quadric forms \eqref{quadric-form}
\begin{equation}
    \begin{split}
        m_{rs}^{(0)}&=\sum_{k,l=1}^{3}(\omega^{*})^{k-1}\omega^{l-1}\mathcal{M}(\omega^{l-1}\epsilon_{s},(\omega^{k-1}\epsilon_{r})^{*})\\
        &=9\frac{\kappa^{[0]*}}{\kappa^{[0]}-\kappa^{[0]*}}{\rm e}^{\vartheta^{[0]}+\vartheta^{[0]*}}\sum_{k=1,l=1}^{\infty,\infty}M_{3k-1,3l-1}\epsilon_{r}^{*(3k-1)}\epsilon_{s}^{3l-1}
    \end{split}
\end{equation}
and
\begin{equation}
    \begin{split}
        m_{rs}^{(1)}&=\sum_{k,l=1}^{3}(\omega^{*})^{k-1}\omega^{l-1}\mathcal{G}(\omega^{l-1}\epsilon_{s},(\omega^{k-1}\epsilon_{r})^{*})\\
        &=9\frac{\kappa^{[0]*}+b_i}{\kappa^{[0]}+b_i}\frac{\kappa^{[0]}}{\kappa^{[0]}-\kappa^{[0]*}}{\rm e}^{\vartheta^{[0]}+\vartheta^{[0]*}}\sum_{k=1,l=1}^{\infty,\infty}G_{3k-1,3l-1}\epsilon_{r}^{*(3k-1)}\epsilon_{s}^{3l-1}.
    \end{split}
\end{equation}
As we discussed earlier, the quadratic forms only depend on the $(3k-1,3l-1)$th order coefficients
with respect to $\epsilon_{s}$ and $\epsilon^{*}_{r}$. Similarly, taking $|y_{s}\rangle=|y_{s}^{(1)}\rangle$ 
and $c^{(1)}_{s}(\epsilon_{s})=c_{s}(\epsilon_{s})$ for all $s=1,2,\cdots,N$, 
the quadratic forms only depend on the $(3k-2,3l-2)$th order coefficients. 
Due to the fact that a solution of the CFL equations \eqref{CFL} multiplied by a constant with a modulus of one 
is still a solution of it, we can discard the terms 
$9\frac{\kappa^{[0]*}}{\kappa^{[0]}-\kappa^{[0]*}}{\rm e}^{\vartheta^{[0]}+\vartheta^{[0]*}}$ in $m_{rs}^{(0)}$ and 
$9\frac{\kappa^{[0]*}+b_i}{\kappa^{[0]}+b_i}\frac{\kappa^{[0]}}{\kappa^{[0]}-\kappa^{[0]*}}{\rm e}^{\vartheta^{[0]}+\vartheta^{[0]*}}$ in $m_{rs}^{(1)}$. 
Take limit $\epsilon_{s}, \epsilon^{*}_{r}\to 0$ in \eqref{N-sol}, for $k=0,1$, 
it leads to $k$-type rogue wave solutions 
\begin{equation}\label{special-rw1}
    u_{i,0}^{[N]}=a_{i}\left(\frac{\det(\left(G_{3k-1,3l-1}\right)_{1\leq k,l\leq N})}{\det(\left(M_{3k-1,3l-1}\right)_{1\leq k,l\leq N})}\right){\rm e}^{{\rm i}\omega_i}, \quad
    u_{i,1}^{[N]}=a_{i}\left(\frac{\det(\left(G_{3k-2,3l-2}\right)_{1\leq k,l\leq N})}{\det(\left(M_{3k-2,3l-2}\right)_{1\leq k,l\leq N})}\right){\rm e}^{{\rm i}\omega_i}.
\end{equation}
For 0-type rogue wave solutions, the free internal parameters are $(\chi^{[1]},\chi^{[2]},\chi^{[4]},\chi^{[5]},\cdots,\chi^{[3N-1]})$ and 
for 1-type rogue wave solutions, the free internal parameters are $(\chi^{[1]},\chi^{[2]},\chi^{[4]},\chi^{[5]},\cdots,\chi^{[3N-2]})$. 
We will study the rogue wave patterns for these two cases. 

Actually, we can generate multi-rogue wave solutions by taking $|y_{s}\rangle=|y_{s}^{(0)}\rangle$ for $1\leq s \leq N_{1}$ and 
$|y_{s}\rangle=|y_{s}^{(1)}\rangle$ for $N_{1}+1\leq s \leq N$ where $N_{1}$ is an integer. Then there will 
occur $(3k-1,3l-2)$th and $(3k-2,3l-1)$th order coefficients. In this case, let $\chi_{l}^{[i]}\in\mathbb{C}$, we consider
$c_{s}^{(l)}(\epsilon_{s})={\rm e}^{\sum_{i=1}^{\infty}\chi_{l}^{[i]}\epsilon_{s}^{i}}$ and $\chi_{l}^{[3i]}=0, i\geq 1$. Then we 
can expand $\vartheta^{(l)}(\epsilon_{s})=\theta(\epsilon_{s})+\sum_{i=1}^{\infty}\chi_{l}^{[i]}\epsilon_{s}^{i}$ just like \eqref{exp-theta}. 
After taking the limit $\epsilon_{s},\epsilon^{*}_{r}\to 0$ in \eqref{N-sol}, we obtain the multi-rogue wave solutions 
\begin{equation}\label{gene-rw1}
    u_i^{(N_{1},N_{2})}=a_{i}\frac{\det
    \begin{pmatrix}
        \mathbf{M}^{(1)}_{N_{1},N_{1}}& \mathbf{M}^{(1)}_{N_{1},N_{2}}\\
        \mathbf{M}^{(1)}_{N_{2},N_{1}}& \mathbf{M}^{(1)}_{N_{2},N_{2}}
    \end{pmatrix}
        }{\det
        \begin{pmatrix}
            \mathbf{M}^{(0)}_{N_{1},N_{1}}& \mathbf{M}^{(0)}_{N_{1},N_{2}}\\
            \mathbf{M}^{(0)}_{N_{2},N_{1}}& \mathbf{M}^{(0)}_{N_{2},N_{2}}
        \end{pmatrix}
        }{\rm e}^{{\rm i}\omega_i},
\end{equation}
where $N_{2}=N-N_{1}$ and the matrices are
\begin{equation}
    \mathbf{M}^{(0)}_{N_{p},N_{q}}=
    \left(M_{3k-p,3l-q}\right)_{1\leq k \leq N_{p},1 \leq l \leq N_{2}},\quad 
    \mathbf{M}^{(1)}_{N_{p},N_{q}}=
    \left(G_{3k-p,3l-q}\right)_{1\leq k \leq N_{p},1 \leq l \leq N_{2}}. 
\end{equation}
For the rogue wave solutions \eqref{gene-rw1}, 
the internal parameters are $(\chi_{0}^{[1]},\chi_{0}^{[2]},\chi_{0}^{[4]},\chi_{0}^{[5]},\cdots,\chi_{0}^{[3N_{1}-1]})$ in $|y_{s}^{(0)}\rangle$ and 
$(\chi_{1}^{[1]},\chi_{1}^{[2]},\chi_{1}^{[4]},\chi_{1}^{[5]},\cdots,\chi_{1}^{[3N_{1}-2]})$ in $|y_{s}^{(1)}\rangle$. If $N_{1}=0$ or $N_{2}=0$, then 
the multi-rogue wave solutions degenerate to $k$-type rogue wave solutions \eqref{special-rw1}. 
To analyze the rogue wave patterns, we need another form of the rogue wave solutions \eqref{gene-rw1}. 
To simplify the notations, we consider the $k$-type rogue wave solutions \eqref{special-rw1} firstly. 
Using the skill $\kappa={\rm{e}}^{\ln(\kappa)}$, $M_{k,l}$ and $G_{k,l}$ can be expressed by Schur polynomials. 
Expanding
\begin{equation*}
    \begin{split}
        \ln(\kappa)&=\ln(\kappa^{[0]})+
        \ln\left(1-\frac{2z^{[1]}}{\kappa^{[0]}}\sum_{i=1}^{\infty}\mu_i\epsilon_{s}^{i}\right)   \\  &=\ln(\kappa^{[0]})-\sum_{j=1}^{\infty}
        \frac{1}{j}\left(\frac{2z^{[1]}}{\kappa^{[0]}}\right)^{j}\left(\sum_{i=1}^{\infty}\mu_i\epsilon_{s}^{i}\right)^{j} \\
        &=\ln(\kappa^{[0]})+\sum_{j=1}^{\infty}H_{j}^{(1)}\epsilon_{s}^{j},
    \end{split}
\end{equation*}
similarly, we have Taylor expansions
\begin{equation}\label{coeff}
\begin{split}
    \ln\left(\frac{\kappa}{\kappa^{[0]}}\right)=\sum_{j=1}^{\infty}H_{j}^{(1)}\epsilon_{s}^{j}
        ,\quad
        &\ln\left(\frac{\kappa^{[0]}-\kappa^{[0]*}}{-2z^{[1]}\mu_{1}\epsilon_{s}}\frac{\kappa-\kappa^{[0]}}{\kappa-\kappa^{[0]*}}\right)=\sum_{j=1}^{\infty}H^{(2)}_{j}\epsilon_{s}^{j},
        \\
        \ln\left(\frac{\kappa^{[0]}-\kappa^{[0]*}}{\kappa-\kappa^{[0]*}}\right)=\sum_{j=1}^{\infty}H^{(3)}_{j}\epsilon_{s}^{j}
        ,\quad
        &\ln\left(\frac{\kappa^{[0]}+b_i}{\kappa+b_i}\right)=\sum_{j=1}^{\infty}H^{(4)}_{j}\epsilon_{s}^{j}.
\end{split}
\end{equation}  
Here we denote $H^{(i)}=(H^{(i)}_{1},H^{(i)}_{2},\cdots),i=1,2,3,4$. Now we can reduce the functions \eqref{denom-expand} 
to more useful forms. Firstly, we expand
\begin{equation}
    \frac{1}{\kappa-\kappa^*}=\frac{\kappa^{[0]*}-\kappa^{[0]}}{(\kappa-\kappa^{[0]*})(\kappa^{*}-\kappa^{[0]})}
    \frac{1}{1-\frac{(\kappa-\kappa^{[0]})(\kappa^{*}-\kappa^{[0]*})}{(\kappa-\kappa^{[0]*})(\kappa^{*}-\kappa^{[0]})}}
    =\frac{\kappa^{[0]*}-\kappa^{[0]}}{(\kappa-\kappa^{[0]*})(\kappa^{*}-\kappa^{[0]})}
    \sum_{j=0}^{\infty}\left(\frac{(\kappa-\kappa^{[0]})(\kappa^{*}-\kappa^{[0]*})}{(\kappa-\kappa^{[0]*})(\kappa^{*}-\kappa^{[0]})}\right)^{j},
\end{equation}
then
the functions \eqref{denom-expand} can be expressed by
\begin{equation}
    \begin{split}
        \mathcal{M}(\epsilon_{s},\epsilon_{r}^{*})
        &=\frac{\kappa^{[0]*}}{\kappa^{[0]}-\kappa^{[0]*}}{\rm e}^{\vartheta^{[0]}+\vartheta^{[0]*}}\sum_{k=0}^{\infty}\left(\frac{4z^{[1]}z^{[1]*}\mu_{1}\mu_{1}^{*}\epsilon_{s}\epsilon_{r}^{*}}{|\kappa^{[0]}-\kappa^{[0]*}|^{2}}\right)^{k}\\
        &\exp\left(\sum_{l=1}^{\infty}((kH^{(2)}_{l}+H_{l}^{(3)}+\vartheta^{[l]})\epsilon_{s}^{l}+(kH^{(2)*}_{l}+H_{l}^{(3)*}+H_{l}^{(1)*}+\vartheta^{[l]*})\epsilon_{r}^{*l})\right),\\
        \mathcal{G}(\epsilon_{s},\epsilon_{r}^{*})
        &=\frac{\kappa^{[0]*}+b_i}{\kappa^{[0]}+b_i}\frac{\kappa^{[0]}}{\kappa^{[0]}-\kappa^{[0]*}}{\rm e}^{\vartheta^{[0]}+\vartheta^{[0]*}}\sum_{k=0}^{\infty}\left(\frac{4z^{[1]}z^{[1]*}\mu_{1}\mu_{1}^{*}\epsilon_{s}\epsilon_{r}^{*}}{|\kappa^{[0]}-\kappa^{[0]*}|^{2}}\right)^{k}\\
        &\exp\left(\sum_{l=1}^{\infty}((kH^{(2)}_{l}+H_{l}^{(3)}+H_{l}^{(4)}+H^{(1)}_{l}+\vartheta^{[l]})\epsilon_{s}^{l}+(kH^{(2)*}_{l}+H_{l}^{(3)*}-H_{l}^{(4)*}+\vartheta^{[l]*})\epsilon_{r}^{*l})\right).
    \end{split}
\end{equation}
Denote $\varTheta=(\vartheta^{[1]},\vartheta^{[2]},\cdots)$, then the coefficients are given by
\begin{equation}\label{coeff-function}
    \begin{split}
        M_{k,l}&=\sum_{r=0}^{\min(k,l)}C^{r}S_{l-r}(rH^{(2)}+H^{(3)}+\varTheta^{})S_{k-r}(rH^{(2)*}+H^{(3)*}+H^{(1)*}+\varTheta^{*}),
        \\
        G_{k,l}&=\sum_{r=0}^{\min(k,l)}C^{r}S_{l-r}(rH^{(2)}+H^{(3)}+H^{(4)}+H^{(1)}+\varTheta^{})S_{k-r}(rH^{(2)*}+H^{(3)*}-H^{(4)*}+\varTheta^{*}),
    \end{split}
\end{equation}
where the constant $C=\frac{4|z^{[1]}|^{2}|\mu_{1}|^{2}}{|\kappa^{[0]}-\kappa^{[0]*}|^{2}}$. Hence for 0-type rogue wave 
solutions in \eqref{special-rw1}, the coefficients
\begin{equation}
    M_{3k-1,3l-1}=
    \begin{pmatrix}
        S_{3k-1}(H^{(3)}+H^{(1)}+\varTheta) \\ C^{\frac{1}{2}}S_{3k-2}(H^{(2)}+H^{(3)}+H^{(1)}+\varTheta) \\ \vdots
    \end{pmatrix}_{3N\times 1}^{\dagger}
    \begin{pmatrix}
        S_{3l-1}(H^{(3)}+\varTheta) \\
        C^{\frac{1}{2}}S_{3l-2}(H^{(2)}+H^{(3)}+\varTheta) \\
        \vdots
    \end{pmatrix}_{3N\times 1}, 
\end{equation}
where $1\leq k,l \leq N$. 
The expression of $G_{3k-1,3l-1}$ is similar to $M_{3k-1,3l-1}$. Hence the coefficients $M_{k,l}$ and $G_{k,l}$ can be expressed 
by Schur polynomials. For the multi-rogue wave solutions \eqref{gene-rw1}, using the same method, 
we can obtain a similar expression. 

To express the multi-rogue wave solutions \eqref{gene-rw1} using Schur polynomials, we introduce 
some notations. 
For two integers $0\leq N_{1},N_{2}\leq N$ and $N_{1}+N_{2}=N$, given two 
infinite dimensional vector $\mathbf{x}^{(r)}=(x_{1}^{(r)},x_{2}^{(r)},x_{3}^{(r)},\cdots),r=1,2$, and $p,q\in\{1,2\}$ we define 
\begin{equation}
    \mathbf{Y}_{N_{p},N_{q}}^{(p,q)}(\mathbf{x}^{(1)},\mathbf{x}^{(2)})=(\mathbf{Y}_{N_{p}}^{(p)}(\mathbf{x}^{(2)}))^{\dagger}\mathbf{Y}_{N_{q}}^{(q)}(\mathbf{x}^{(1)}),
\end{equation}
where $\mathbf{Y}_{N_{p}}^{(p)}(\mathbf{x}^{(r)})$ is $3N$ by $N_{p}$ matrix, and the $(i,j)$th element is
\begin{equation}\label{m-CFL-element}
    Y^{(p)}_{N_{p};i,j}(\mathbf{x}^{(r)})=C^{\frac{i-1}{2}}S_{3j-i-p+1}((i-1)H^{(2)}+\varTheta_{p}+\mathbf{x}^{(r)}).
\end{equation}
The terms $\varTheta_{p}=(\vartheta_{p-1}^{[1]},\vartheta_{p-1}^{[2]},\cdots)$, and the coefficients $\vartheta_{p-1}^{[i]}$ 
are given by the expansion $\vartheta^{(p-1)}(\epsilon_{s})=\theta(\epsilon_{s})+\sum_{i=1}^{\infty}\chi_{p-1}^{[i]}\epsilon_{s}^{i}
=\sum_{i=0}^{\infty}\vartheta_{p-1}^{[i]}\epsilon_{s}^{i}$ just like \eqref{exp-theta}. 
Then the $k$-type rogue wave solutions \eqref{special-rw1} can be represented by
\begin{equation}\label{special-rw2}
    \begin{split}
        u_{i,0}^{[N]}=a_{i}\left(\frac{\det(\mathbf{Y}_{N,N}^{(1,1)}(H^{(3)}+H^{(4)}+H^{(1)},H^{(3)}-H^{(4)}))}{\det(\mathbf{Y}_{N,N}^{(1,1)}(H^{(3)},H^{(3)}+H^{(1)}))}\right){\rm e}^{{\rm i}\omega_i}, \\
        u_{i,1}^{[N]}=a_{i}\left(\frac{\det(\mathbf{Y}_{N,N}^{(2,2)}(H^{(3)}+H^{(4)}+H^{(1)},H^{(3)}-H^{(4)}))}{\det(\mathbf{Y}_{N,N}^{(2,2)}(H^{(3)},H^{(3)}+H^{(1)}))}\right){\rm e}^{{\rm i}\omega_i}.
    \end{split}
\end{equation}
For the multi-rogue wave solutions \eqref{gene-rw1}, 
denote 
\begin{equation}
    \mathbf{Y}_{N_{1},N_{2}}(\mathbf{x}^{(1)},\mathbf{x}^{(2)})=
    \begin{pmatrix}
        \mathbf{Y}_{N_{1},N_{1}}^{(1,1)} & \mathbf{Y}_{N_{1},N_{2}}^{(1,2)} \\
        \mathbf{Y}_{N_{2},N_{1}}^{(2,1)} & \mathbf{Y}_{N_{2},N_{2}}^{(2,2)} 
    \end{pmatrix},
\end{equation}
we have the following proposition about the multi-rogue wave solutions of the CFL equations \eqref{CFL}. 
\begin{prop}
    Given two integers $N_{1},N_{2}$ with $0\leq N_{1},N_{2}\leq N$ and $N_{1}+N_{2}=N$, 
    let $|y_{s}\rangle=|y_{s}^{(0)}\rangle$ for $1\leq s \leq N_{1}$ and 
    $|y_{s}\rangle=|y_{s}^{(1)}\rangle$ for $N_{1}+1\leq s \leq N$ 
    in Theorem \ref{DT} and the seed solutions $u_{i}=u_{i}^{[0]}$ in \eqref{seed}. By B\"acklund transformation, 
    the CFL equations \eqref{CFL} have multi-rogue wave solutions
    \begin{equation}\label{gene-rw2}
        u_i^{(N_{1},N_{2})}=a_{i}\left(\frac{\det(\mathbf{Y}_{N_{1},N_{2}}(H^{(3)}+H^{(4)}+H^{(1)},H^{(3)}-H^{(4)}))}{\det(\mathbf{Y}_{N_{1},N_{2}}(H^{(3)},H^{(3)}+H^{(1)}))}\right){\rm e}^{{\rm i}\omega_i}.
    \end{equation}
    For the multi-rogue wave solutions \eqref{gene-rw2}, the free internal parameters are $(\chi_{0}^{[1]},\chi_{0}^{[2]},\chi_{0}^{[4]},\chi_{0}^{[5]},\cdots,\chi_{0}^{[3N_{1}-1]})$ for $|y_{s}^{(0)}\rangle$
    in \eqref{ys0-discussion} and 
    $(\chi_{1}^{[1]},\chi_{1}^{[2]},\chi_{1}^{[4]},\chi_{1}^{[5]},\cdots,\chi_{1}^{[3N_{1}-2]})$ for $|y_{s}^{(1)}\rangle$ in \eqref{ys1-discussion}. 
\end{prop}
The 0-type rogue wave solutions $u_{i,0}^{[N]}$ in \eqref{special-rw2} have the form $u_{i}^{(N,0)}$ in \eqref{gene-rw2} and 
the 1-type rogue wave solutions $u_{i,1}^{[N]}$ have the form $u_{i}^{(0,N)}$. For these two cases, we use the notations 
$(\chi^{[1]},\chi^{[2]},\chi^{[4]},\chi^{[5]},\cdots,\chi^{[3N-1]})$ and 
$(\chi^{[1]},\chi^{[2]},\chi^{[4]},\chi^{[5]},\cdots,\chi^{[3N-2]})$ respectively to 
represent the free internal parameters. 
Since the selections of parameters $\chi^{[3i]},i\geq 1$ do not impact the rogue wave solution, we can set these terms to zero. 
We will provide the reasons behind the proof of the rogue wave patterns in the inner region. 

In the next section, we will analyze the $k$-type rogue wave solutions \eqref{special-rw2}, and show their patterns. 
Actually, there are three types of rogue wave solutions. If we take $|y_{s}\rangle=\mathbf{\Phi}(\lambda_{s})(c_{s}(\epsilon_{s}), c_{s}(\omega\epsilon_{s}), c_{s}(\omega^{2}\epsilon_{s}))^{T}$, then 
the quadratic forms \eqref{quadric-form} only depends on $(3k,3l)$th order coefficients of two functions 
\eqref{denom-expand}. In this case, the solution can be converted to $0$-type, since 
the first column of $\mathbf{Y}_{N_{p}}^{(p)}(\mathbf{x}^{(r)})$ \eqref{m-CFL-element} has only one nonzero element $S_{0}=1$. 

As a summary of this section, considering the seed solutions \eqref{seed}, 
it leads to the fundamental solution \eqref{fund2} of the Lax pair \eqref{CFL-lax} 
and the Riemann surfaces \eqref{RS1}. Additionally, we study a general proposition of the Riemann surfaces at branch points. Then
we construct the rogue wave solutions \eqref{gene-rw2} generated at the branch point of multiplicity three 
using the B\"acklund transformation. 
To analyze the rogue wave patterns in Section \ref{sec-pattern}, we reduce 
the multi-rogue wave solutions to determinant 
representation \eqref{gene-rw2}. 
In the next section, we use the root structures of Okamoto polynomial hierarchies to study 
$k$-type rogue wave solutions \eqref{special-rw2}. 

\section{The rogue wave patterns}\label{sec-pattern}

In this section, we study the rogue wave patterns
for \eqref{special-rw2}, and our results are as follows. 
Under the assumption nonzero roots of Okamoto polynomial hierarchies are all simple, the patterns are divided into two parts, 
the outer region, and the inner region. In the outer region, the rogue wave can be decomposed into some
first-order rogue wave solutions which are far from the origin. In the inner region, it can be viewed as a lower-order rogue wave. 
The positions and orders of these rogue waves are associated with the root distributions of Okamoto polynomial hierarchies respectively. 

\subsection{The asymptotics of the outer region}
Now we study the rogue wave patterns in the outer region for \eqref{special-rw2}. 
\begin{prop}[Outer region]
    Let $\eta=(\chi^{[m]})^{1/m}, m\geq 2$ where $\chi^{[m]}$ is an internal parameter of the $k$-type rogue wave solutions 
    \eqref{special-rw2} and $\sqrt{x^2+t^2}=\mathcal{O}(\eta)$, suppose the nonzero roots of $W_{N}^{[k,m]}(z)$ 
    are all simple. As $|\eta|\to \infty$ for $k$-type rogue wave solutions \eqref{special-rw2} with $k=0,1$, we have first-order rogue wave solutions 
    for $i=1,2$ near the nonzero roots $(x_{0},t_{0})$ of $W_{N}^{[k,m]}((\vartheta^{[1]}(x,t)-\chi^{[1]}){\rm e}^{-{\rm i}\arg{\eta}})$: 
    \begin{equation}\label{outer-solution}
        u_{i,k}^{asy}(\hat{x},\hat{t})=
        a_{i}
        \left(
        \frac
        {|p_{1}|^{2}(\hat{x}+\frac{\Re{(p_{1}q_{1}^{*})}}{|p_{1}|^{2}}\hat{t}+r_{3})^{2}+p_{2} (\hat{t}+r_{4})^{2}+\frac{4|z^{[1]}|^{2}|\mu_{1}|^{2}}{|\kappa^{[0]}-\kappa^{[0]*}|^{2}}}
        {|p_{1}|^{2}(\hat{x}+\frac{\Re{(p_{1}q_{1}^{*})}}{|p_{1}|^{2}}\hat{t}+r_{1})^{2}+p_{2} (\hat{t}+r_{2})^{2}+\frac{4|z^{[1]}|^{2}|\mu_{1}|^{2}}{|\kappa^{[0]}-\kappa^{[0]*}|^{2}}}
        \right)
        {\rm e}^{{\rm i}\omega_i}+\mathcal{O}(\eta^{-1}), 
    \end{equation}
    where $\hat{x}=x-x_{0}|\eta|,\hat{t}=t-t_{0}|\eta|$. The translation terms are
    \begin{equation}
        \begin{split}
            r_{1}&=\Re(\frac{q_{2}}{p_{1}})+\frac{H_{1}^{(1)*}}{2p_{1}^{*}},\quad
            r_{2}=-\frac{\Im{(p_{1}q_{1}^{*})}}{2p_{2}|p_{1}|^{2}}(2\Im(q_{2} p_{1}^{*})+{\rm i}H_{1}^{(1)*}p_{1}),\\
            r_{3}&=\Re(\frac{q_{2}}{p_{1}})+{\rm i}\Im(\frac{H_{1}^{(4)}}{p_{1}})+\frac{H^{(1)}_{1}}{2p_{1}},\quad
            r_{4}=-\frac{\Im{(p_{1}q_{1}^{*})}}{2p_{2}|p_{1}|^{2}}(2\Im(q_{2} p_{1}^{*})-2{\rm i}\Re(H_{1}^{(4)}p_{1}^{*})-{\rm i}H^{(1)}_{1}p_{1}^{*}),\\
        \end{split}
    \end{equation}
    where 
    \begin{equation}
        \begin{split}
            p_{1}&=-2{\rm i} z^{[1]}\mu_{1},\quad
            q_{1}=-\frac{{\rm i} z^{[1]}(b_{1}+b_{2}+2\kappa^{[0]}-2z^{[0]})}{b_{1}b_{2}z^{[0]}},\\
            p_{2}&=|q_{1}|^{2}-(\frac{\Re{(p_{1}q_{1}^{*})}}{|p_{1}|})^{2},\quad
            q_{2}=H^{(3)}_{1}+\chi^{[1]}+(\vartheta^{[2]}(x_{0},t_{0})-\chi^{[2]}){\rm e}^{-{\rm i}\arg{\eta}}\frac{W_{N,1}^{[k,m]}((\vartheta^{[1]}(x_{0},t_{0})-\chi^{[1]}){\rm e}^{-{\rm i}\arg{\eta}})}{W_{N}^{[k,m]'}((\vartheta^{[1]}(x_{0},t_{0})-\chi^{[1]}){\rm e}^{-{\rm i}\arg{\eta}})}.
        \end{split}
    \end{equation}
\end{prop} 
\begin{proof}
We will only provide the proof for 1-type as the proofs are similar to 0-type. 
Our main idea is to estimate the determinant element 
of the numerator and denominator in \eqref{special-rw2},i.e. 
$\mathbf{Y}_{0,N}(\mathbf{x}^{(1)},\mathbf{x}^{(2)})$, where $(\mathbf{x}^{(1)},\mathbf{x}^{(2)})=(H^{(3)},H^{(3)}+H^{(1)})$ or 
$(\mathbf{x}^{(1)},\mathbf{x}^{(2)})=(H^{(3)}+H^{(4)}+H^{(1)},H^{(3)}-H^{(4)})$. Since $\mathbf{Y}_{0,N}(\mathbf{x}^{(1)},\mathbf{x}^{(2)})$ 
can be expressed by $\mathbf{Y}_{N}^{(2)}$ \eqref{m-CFL-element}, 
it just needs to estimate 
$\mathbf{Y}_{N}^{(2)}(\mathbf{x})$ for a given vector $\mathbf{x}=(x_1,x_2,\cdots)$. 
By Proposition \ref{prop-OK2}, for $i\geq 2$ we obtain
\begin{equation}
    S_{i}(kH^{(2)}+\varTheta^{}+\mathbf{x})= S_{i}(\mathbf{v}_{1})+
    \begin{cases}
        \mathcal{O}(\eta^{i-1}), \quad m\geq 3,\\
        \mathcal{O}(\eta^{i-2}), \quad m=2
    \end{cases}
\end{equation}
for some integer $k$, where $\mathbf{v}_{1}=(\vartheta^{[1]}+x_{1},0,\cdots,0,\vartheta^{[m]}+kH^{(2)}_{m}+x_m,0,\cdots)$, since $H_{1}^{(2)}=0$. 
By Proposition \ref{prop-Ok1}, since $\vartheta^{[1]}(x,t)=\mathcal{O}(\eta)$ and $x_{1},H^{(2)}_{m},x_m$ are constants, 
we have $S_{i}(\mathbf{v}_{1})=\eta^{i}p_{i}^{[m]}(\eta^{-1}\vartheta^{[1]})+\mathcal{O}(\eta^{i-1})$. 
Using Okamoto polynomial hierarchies, it leads to
\begin{equation}
    \det_{1\leq i,j \leq N}(S_{3j-i-1}((i-1)H^{(2)}+\varTheta^{}+\mathbf{x}))=\eta^{N^{2}}(c_{N}^{[1]})^{-1}W_{N}^{[1,m]}(\eta^{-1}\vartheta^{[1]})
    +\mathcal{O}(\eta^{N^{2}-1}),
\end{equation}
where $c_{N}^{[1]}$ is defined in \eqref{Okamoto}. 
To calculate the asymptotic expression of the numerator and denominator in the rogue wave solutions \eqref{special-rw2}, 
using the Cauchy-Binet formula, we obtain
\begin{equation}\label{est}
    \begin{split}
        \det(\mathbf{Y}_{0,N}(\mathbf{x}^{(1)},\mathbf{x}^{(2)}))
        &=\sum_{1\leq v_{1}<v_{2}<\cdots <v_{N}\leq 3N}\det_{1\leq i,j \leq N}(Y^{(2)}_{N,v_{j},i}(\mathbf{x}^{(2)})^{\dagger})\det_{1\leq i,j \leq N}(Y^{(2)}_{N,v_{j},i}(\mathbf{x}^{(1)}))\\
        &=\sum_{1\leq v_{1}<v_{2}<\cdots <v_{N}\leq 3N}C^{-\sum_{i=1}^{N}v_{i}}
        \det_{1\leq i,j \leq N}(S^{*}_{3i-v_{j}-1}((v_{j}-1)H^{(2)}+\mathbf{x}^{(2)}
        \\ & +\varTheta^{}))\det_{1\leq i,j \leq N}(S_{3i-v_{j}-1}((v_{j}-1)H^{(2)}+\mathbf{x}^{(1)}+\varTheta^{})).\\
    \end{split}
\end{equation}
Since the degree of $S_{i}(\mathbf{v}_{1})$ is decrease with respect to $\eta$ when $i$ is decrease, 
the leading order term of $\eta$ comes from the choice 
$(v_{1},\cdots,v_{N})=(1,2,\cdots, N)$. Hence the coefficient of the leading order term is
\begin{equation}\label{asym}
    C^{-\frac{N(N+1)}{2}}\det_{1\leq i,j \leq N}(S^{*}_{3i-j-1}((j-1)H^{(2)}+\mathbf{x}^{(2)}+\varTheta^{}))
    \det_{1\leq i,j \leq N}(S_{3i-j-1}((j-1)H^{(2)}+\mathbf{x}^{(1)}+\varTheta^{})),
\end{equation}
which has an asymptotic expansion
\begin{equation}
    C^{-\frac{N(N+1)}{2}}|\eta|^{2N^{2}}\left|(c_{N}^{[1]})^{-1}W_{N}^{[1,m]}(\eta^{-1}\vartheta^{[1]})\right|^{2}.
\end{equation}
Under the condition $\sqrt{x^2+t^2}=\mathcal{O}(\eta)$, if $\eta^{-1}\vartheta^{[1]}(x,t)$ 
is far from the roots of the Okamoto polynomial hierarchies, 
when $|\eta|\to \infty$, we obtain its limit is nonzero and independent of $\mathbf{x}^{(1)}$ and $\mathbf{x}^{(2)}$.
Hence the asymptotic solution of \eqref{special-rw2} just $a_{i}{\rm e}^{{\rm i}\omega_i}$ in this case.  
To get a nontrivial asymptotic expansion, we take a nonzero root $(x_{0},t_{0})$ of 
$W_{N}^{[1,m]}((\vartheta^{[1]}(x,t)-\chi^{[1]}){\rm e}^{-{\rm i}\arg{\eta}})$ and expand \eqref{special-rw2} near 
$(x_{0},t_{0})$. 

We first calculate the leading order term that comes from the choice 
$(v_{1},\cdots,v_{N})=(1,2,\cdots, N)$. 
Making coordinate transformation $x=\hat{x}+x_{0}\eta {\rm e}^{-{\rm i}\arg{\eta}}, t=\hat{t}+t_{0}\eta {\rm e}^{-{\rm i}\arg{\eta}}$ 
(Note that $\eta {\rm e}^{-{\rm i}\arg{\eta}}=|\eta|$ is real, then the transformation is reasonable), if $\chi^{[1]}=0$, we have
\begin{equation}
    \begin{split}
        &\exp((\frac{\vartheta^{[1]}(\hat{x},\hat{t})+x_{1}}{\eta}+\vartheta^{[1]}(x_{0},t_{0}){\rm e}^{-{\rm i}\arg{\eta}})\epsilon+(\frac{\vartheta^{[2]}(\hat{x},\hat{t})+x_{2}}{\eta^{2}}+\frac{(\vartheta^{[2]}(x_{0},t_{0})-\chi^{[2]}){\rm e}^{-{\rm i}\arg{\eta}}}{\eta})\epsilon^{2}+\epsilon^{m})\\
        =&\exp(\vartheta^{[1]}(x_{0},t_{0}){\rm e}^{-{\rm i}\arg{\eta}}\epsilon+\epsilon^{m})\exp(\frac{\vartheta^{[1]}(\hat{x},\hat{t})+x_{1}}{\eta}\epsilon+\frac{(\vartheta^{[2]}(x_{0},t_{0})-\chi^{[2]}){\rm e}^{-{\rm i}\arg{\eta}}}{\eta}\epsilon^{2}+\mathcal{O}(\eta^{-2}))\\
        =&(\sum_{j=1}^{\infty}p_{j}^{[m]}(\vartheta^{[1]}(x_{0},t_{0}){\rm e}^{-{\rm i}\arg{\eta}})\epsilon^{j})(1+\frac{\vartheta^{[1]}(\hat{x},\hat{t})+x_{1}}{\eta}\epsilon+\frac{(\vartheta^{[2]}(x_{0},t_{0})-\chi^{[2]}){\rm e}^{-{\rm i}\arg{\eta}}}{\eta}\epsilon^{2}+\mathcal{O}(\eta^{-2})).\\
    \end{split}
\end{equation}
Hence we have an approximation
\begin{equation}
\begin{split}
    \eta^{-n}S_{n}(kH^{(2)}+\varTheta^{}+\mathbf{x})\sim & p_{n}^{[m]}(\vartheta^{[1]}(x_{0},t_{0}){\rm e}^{-{\rm i}\arg{\eta}})\\
    &+\left((\vartheta^{[1]}(\hat{x},\hat{t})+x_{1})p_{n-1}^{[m]}(\vartheta^{[1]}(x_{0},t_{0}){\rm e}^{-{\rm i}\arg{\eta}})
    \right.\\
    &\left.+(\vartheta^{[2]}(x_{0},t_{0})-\chi^{[2]}){\rm e}^{-{\rm i}\arg{\eta}}p_{n-2}^{[m]}(\vartheta^{[1]}(x_{0},t_{0}){\rm e}^{-{\rm i}\arg{\eta}})\right)
    \eta^{-1}+\mathcal{O}(\eta^{-2}).
\end{split}
\end{equation}
If $\chi^{[1]}$ is not zero, we just replace $\vartheta^{[1]}(x_{0},t_{0})$ by $\vartheta^{[1]}(x_{0},t_{0})-\chi^{[1]}$. 
Using the properties of determinants, 
the coefficient of $\eta^{N^2}$ for $\det(\mathbf{Y}_{N}^{(2)}(\mathbf{x}))$ is zero since $(x_{0},t_{0})$ is the 
root of $W_{N}^{[1,m]}((\vartheta^{[1]}(x,t)-\chi^{[1]}){\rm e}^{-{\rm i}\arg{\eta}})$. For the coefficient of $\eta^{N^2-1}$, 
letting $z_{1}=(\vartheta^{[1]}(x_{0},t_{0})-\chi^{[1]}){\rm e}^{-{\rm i}\arg{\eta}}$ and $z_{2}=(\vartheta^{[2]}(x_{0},t_{0})-\chi^{[2]}){\rm e}^{-{\rm i}\arg{\eta}}$, 
we can only choose one column be the elements $(\vartheta^{[1]}(\hat{x},\hat{t})+x_{1})p_{n-1}^{[m]}(z_{1})$ or 
$z_{2}p_{n-2}^{[m]}(z_{1})$. 
Hence the leading order term of $\det(\mathbf{Y}_{N}^{(2)}(\mathbf{x}))$ with respect to $\eta$ is
\begin{equation}\label{contr1}
    \eta^{N^2-1}((\vartheta^{[1]}(\hat{x},\hat{t})+x_{1})(c_{N}^{[1]})^{-1}W_{N}^{[1,m]'}(z_{1})+z_{2}(c_{N}^{[1]})^{-1}W_{N,1}^{[1,m]}(z_{1})).
\end{equation}
Another contribution is given by the choice $(v_{1},\cdots,v_{N})=(1,2,\cdots, N-1,N+1)$ in \eqref{est}, then the term 
$\det_{1\leq i,j \leq N}(S_{3i-v_{j}-1}((v_{j}-1)H^{(2)}+\mathbf{x}+\varTheta^{}))$
is
\begin{equation}\label{contr2}
    \eta^{N^2-1}(c_{N}^{[1]})^{-1}W_{N}^{[1,m]'}(z_{1}).
\end{equation}
Combining \eqref{contr1} and \eqref{contr2} and calculating the leading order term of \eqref{est} with respect to $\eta$, 
as $|\eta|\to \infty$, we obtain the asymptotic 1-type rogue wave solution \eqref{special-rw2}
\begin{equation}\label{outer-solution-1}
    a_{i}\frac{\left(\vartheta^{[1]}(\hat{x},\hat{t})+H^{(3)}_{1}+H^{(4)}_{1}+H^{(1)}_{1}+z_{2}\frac{W_{N,1}^{[1,m]}(z_{1})}{W_{N}^{[1,m]'}(z_{1})}\right)
    \left(\vartheta^{[1]}(\hat{x},\hat{t})+H^{(3)}_{1}-H^{(4)}_{1}+z_{2}\frac{W_{N,1}^{[1,m]}(z_{1})}{W_{N}^{[1,m]'}(z_{1})}\right)^{*}+C}
    {\left(\vartheta^{[1]}(\hat{x},\hat{t})+H^{(3)}_{1}+z_{2}\frac{W_{N,1}^{[1,m]}(z_{1})}{W_{N}^{[1,m]'}(z_{1})}\right)
    \left(\vartheta^{[1]}(\hat{x},\hat{t})+H^{(3)}_{1}+H^{(1)}_{1}+z_{2}\frac{W_{N,1}^{[1,m]}(z_{1})}{W_{N}^{[1,m]'}(z_{1})}\right)^{*}+C}
    {\rm e}^{{\rm i}\omega_i}.
\end{equation}
By simplifying the aforementioned expression \eqref{outer-solution-1}, we arrive at the final expression \eqref{outer-solution} 
stated in the proposition.

\end{proof}
Moreover, it is routine to verify that the center $(\hat{x}_{1},\hat{t}_{1})$ of the rogue wave solution \eqref{outer-solution} is
\begin{equation}\label{center}
    \begin{split}
        \left(-\frac{\Re{\left(p_{1}q_{1}^{*}\right)}\Im{\left(p_{1}q_{1}^{*}\right)}}{2p_{2}|p_{1}|^{4}}\Im\left(2q_{2} p_{1}^{*}-H_{1}^{(1)*}p_{1}\right)-\Re\left(\frac{q_{2}}{p_{1}}+\frac{H_{1}^{(1)*}}{2p_{1}^{*}}\right),\frac{\Im{\left(p_{1}q_{1}^{*}\right)}}{2p_{2}|p_{1}|^{2}}\left(\Im\left(2q_{2} p_{1}^{*}-H_{1}^{(1)*}p_{1}\right)\right)\right).
    \end{split}
\end{equation}

As evidenced by the proof, when one of the internal parameters is large enough, 
Okamoto polynomial hierarchies become inherent, as asserted in Proposition \ref{prop-OK2}. 
Proposition \ref{prop-OK2} provides insight into the values of the rogue wave solutions \eqref{special-rw2} that are distant from the origin.
The values of rogue wave solutions at points close to the roots of the Okamoto polynomial hierarchies 
can be approximated by first-order rogue waves. 

\subsection{The asymptotics of the inner region}
Now we calculate the asymptotic expression of rogue wave solutions \eqref{special-rw2} in the inner region. 
We have the following proposition and 
the proof is relative to Theorem \ref{Ok-root}. 
\begin{prop}[Inner region]
    Let $\sqrt{x^2+t^2}=\mathcal{O}(1)$, denote $N=km+N_{0}$ where $N_{0}$ is the remainder of $N$ divided by $m$ and 
    $m\geq 2$. If $m$ is not a multiple of $3$, 
    as $|\chi^{[m]}|\to \infty$ where $\chi^{[m]}$ is an internal parameter for $k$-type rogue wave solutions in \eqref{special-rw2}, 
    we have lower order asymptotic rogue wave solutions \eqref{gene-rw2} for $i=1,2$:
    \begin{equation}\label{inner-solution}
        u_{i,k}^{(N_{1}^{[k]},N_{2}^{[k]})}=
        a_{i}
        \left(
        \frac{\det(\mathbf{Y}_{N_{1}^{[k]},N_{2}^{[k]}}(H^{(3)}+H^{(4)}+H^{(1)},H^{(3)}-H^{(4)}))}{\det(\mathbf{Y}_{N_{1}^{[k]},N_{2}^{[k]}}(H^{(3)},H^{(3)}+H^{(1)}))}
        \right)
        {\rm e}^{{\rm i}\omega_i}+\mathcal{O}((\chi^{[m]})^{-1}).
    \end{equation}
    The free internal parameters are given by $\chi_{l}^{[i]}=\chi^{[i]}+kmH_{i}^{(2)},l=0,1$ for $i\ne m$ and 
    $\chi_{0}^{[m]}=\chi_{1}^{[m]}=kmH_{m}^{(2)}$. 
\end{prop}

\begin{proof}
We also only provide the proof for $1$-type as the proofs are similar for $0$-type. Similar to the proof in the outer region, 
we will estimate $\mathbf{Y}_{N}^{(2)}(\mathbf{x})$ for $\mathbf{x}=H^{(3)},H^{(3)}+H^{(1)},H^{(3)}+H^{(4)}+H^{(1)},H^{(3)}-H^{(4)}$. 
Using the Cauchy-Binet formula, we also need to calculate the limit of \eqref{est} as $\chi^{[m]}\to\infty$. 
Since the elements of $\mathbf{Y}_{N}^{(2)}(\mathbf{x})$ can be viewed as the polynomial of $\chi^{[m]}$, we 
calculate the coefficients with respect to $\chi^{[m]}$. We will prove that, the matrix $\mathbf{Y}_{N}^{(2)}(\mathbf{x})$ 
can be simplified to the form
\begin{equation}\label{Yasym}
    \begin{split}
        \mathbf{Y}_{N}^{(2)}(\mathbf{x})\sim
        \begin{pmatrix}
            \mathbf{A}_{1,1} & \mathbf{A}_{1,2} & \mathbf{A}_{1,3} & \cdots & \mathbf{A}_{1,k} & \mathbf{B}_{1,k+1} \\
            \mathbf{0}_{m\times m} & \mathbf{A}_{2,2} & \mathbf{A}_{2,3} & \cdots & \mathbf{A}_{2,k} & \mathbf{B}_{2,k+1} \\
            \mathbf{0}_{m\times m} & \mathbf{0}_{m\times m} & \mathbf{A}_{3,3} & \cdots & \mathbf{A}_{3,k} & \mathbf{B}_{3,k+1} \\
            \vdots & \vdots & \vdots  & \ddots & \vdots & \vdots \\
            \mathbf{0}_{m\times m} & \mathbf{0}_{m\times m} & \mathbf{0}_{m\times m} & \cdots & \mathbf{A}_{k,k} & \mathbf{B}_{k,k+1} \\
            \mathbf{0}_{(2N+N_{0})\times m} & \mathbf{0}_{(2N+N_{0})\times m} & \mathbf{0}_{(2N+N_{0})\times m} & \cdots & \mathbf{0}_{(2N+N_{0})\times m} & \mathbf{C}_{k+1,k+1}
        \end{pmatrix}
    \end{split}
\end{equation}
where $\mathbf{C}_{k+1,k+1}=\mathbf{C}_{k+1,k+1}(\mathbf{x})$ is $(2N+N_{0})\times N_{0}$ matrix. 
The $m\times m$ matrix $\mathbf{A}_{i,i},i=1,2,\cdots,k$ can be transformed into an upper triangular matrix 
through column transformation, and the diagonal elements are the powers of $\chi^{[m]}$. 
Hence the term \eqref{est} can be approximated by
\begin{equation}\label{asy-inner}
    \begin{split}
        \det(\mathbf{Y}_{0,N}(\mathbf{x}^{(1)},\mathbf{x}^{(2)}))=&\det((\mathbf{Y}_{N}^{(2)}(\mathbf{x}^{(2)}))^{\dagger}\mathbf{Y}_{N}^{(2)}(\mathbf{x}^{(1)}))\\
        \sim& \det((\mathbf{C}_{k+1,k+1}(\mathbf{x}^{(2)}))^{\dagger}\mathbf{C}_{k+1,k+1}(\mathbf{x}^{(1)}))(\chi^{[m]})^{l}
    \end{split}
\end{equation}
for some integer $l$. 
Then we can obtain the proposition. Now we conduct precise calculations. 

To estimate the asymptotic expression, we need to calculate the leading order terms with respect to $\chi^{[m]}$. 
Denote $\hat{\mathbf{x}}^{(r)}=(r-1)H^{(2)}+\varTheta+\mathbf{x}$ and 
$\hat{\mathbf{y}}^{(r)}=\hat{\mathbf{x}}^{(r)}-\chi^{[m]}\mathbf{e}_{m}$ for some integer $r$, where 
$\mathbf{e}_{m}$ is the $m$-th unit vector, we have
\begin{equation}\label{expand1}
    S_{n}(\hat{\mathbf{x}}^{(r)})=\sum_{i=0}^{[n/m]}\frac{(\chi^{[m]})^{i}}{i!}S_{n-im}(\hat{\mathbf{y}}^{(r)}).
\end{equation}
Now we concentrate on the first $m$ columns of $\mathbf{Y}_{N}^{(2)}(\mathbf{x})$ and obtain $\mathbf{A}_{1,1}$. 
Since $m$ is not a multiple of $3$, the order of $3$ in the cyclic group $Z_{m}$ is $m$. Then for the first $m$ columns
, with respect to the leading order terms coefficient of $\chi^{[m]}$ i.e. the term $S_{n-im}(\mathbf{y})$, the subscript 
$n-im$ traverse $0$ to $m-1$ (we omit the coefficient $1/i!$). On the other hand 
, if the row index is decreased by 1, then the term $S_{n-im}$ be $S_{n-im-1}$. Denote the leading order of $1$-st row and 
$j$-th column with respect to $\chi^{[m]}$ be $k_{1,j}$. We only preserve the coefficient of $(\chi^{[m]})^{k_{1,j}}$ in each column. 
Then we obtain $\mathbf{A}_{1,1}$. For example, if $m=3j+1,j\geq 1$, the first $m$ columns of $\mathbf{Y}_{N}^{(2)}(\mathbf{x})$ are 
$3N\times m$ matrix
\begin{equation}
    \begin{pmatrix}
        S_{1}(\hat{\mathbf{x}}^{(1)})  & \cdots & S_{3j-2}(\hat{\mathbf{x}}^{(1)}) & S_{3j+1}(\hat{\mathbf{x}}^{(1)}) & \cdots & S_{6j+1}(\hat{\mathbf{x}}^{(1)}) & S_{6j+4}(\hat{\mathbf{x}}^{(1)}) & \cdots & S_{9j+1}(\hat{\mathbf{x}}^{(1)}) \\
        S_{0}(\hat{\mathbf{x}}^{(2)})  & \cdots & S_{3j-3}(\hat{\mathbf{x}}^{(2)}) & S_{3j}(\hat{\mathbf{x}}^{(2)}) & \cdots & S_{6j}(\hat{\mathbf{x}}^{(2)}) & S_{6j+3}(\hat{\mathbf{x}}^{(2)}) & \cdots & S_{9j}(\hat{\mathbf{x}}^{(2)}) \\
        \vdots & \ddots & \vdots & \vdots & \ddots & \vdots & \vdots & \ddots & \vdots 
    \end{pmatrix}.
\end{equation}
Then we expand the above elements just like \eqref{expand1}. For example, 
$S_{3j-2}(\hat{\mathbf{x}}^{(1)})=S_{3j-2}(\hat{\mathbf{y}}^{(1)})$, 
$S_{3j+1}(\hat{\mathbf{x}}^{(1)})=S_{3j+1}(\hat{\mathbf{y}}^{(1)})+S_{0}(\hat{\mathbf{y}}^{(1)})\chi^{[m]}$
and 
$S_{9j+1}(\hat{\mathbf{x}}^{(1)})=S_{9j+1}(\hat{\mathbf{y}}^{(1)})+S_{6j}(\hat{\mathbf{y}}^{(1)})\chi^{[m]}$ $+S_{3j-1}(\hat{\mathbf{y}}^{(1)})(\chi^{[m]})^{2}$. 
Then we preserve the 
term $(\chi^{[m]})^{0}$ for the first $j$ columns, with respect to $(\chi^{[m]})^{1}$ for $j+1$ to $2j+1$ columns and 
with respect to $(\chi^{[m]})^{2}$ for $2j+2$ to $3j+1=m$ columns. The first $m$ columns of $\mathbf{Y}_{N}^{(2)}(\mathbf{x})$ can be 
approximated by a $3N\times m$ matrix
\begin{equation}
    \begin{split}
        \begin{pmatrix}
            S_{1}(\hat{\mathbf{y}}^{(1)})  & \cdots & S_{3j-2}(\hat{\mathbf{y}}^{(1)}) & \chi^{[m]} & \cdots & S_{3j}(\hat{\mathbf{y}}^{(1)})\chi^{[m]} & S_{2}(\hat{\mathbf{y}}^{(1)})(\chi^{[m]})^{2} & \cdots & S_{3j-1}(\hat{\mathbf{y}}^{(1)})(\chi^{[m]})^{2} \\
            1  & \cdots & S_{3j-3}(\hat{\mathbf{y}}^{(2)}) & 0 & \cdots & S_{3j-1}(\hat{\mathbf{y}}^{(2)})\chi^{[m]} & S_{1}(\hat{\mathbf{y}}^{(2)})(\chi^{[m]})^{2} & \cdots & S_{3j-2}(\hat{\mathbf{y}}^{(2)})(\chi^{[m]})^{2} \\
            0  & \cdots & S_{3j-4}(\hat{\mathbf{y}}^{(3)}) & 0 & \cdots & S_{3j-2}(\hat{\mathbf{y}}^{(3)})\chi^{[m]} & (\chi^{[m]})^{2} & \cdots & S_{3j-3}(\hat{\mathbf{y}}^{(3)})(\chi^{[m]})^{2} \\
            \vdots & \ddots & \vdots & \vdots & \ddots & \vdots & \vdots & \ddots & \vdots 
        \end{pmatrix}
    \end{split}
\end{equation}
since $S_{0}(\hat{\mathbf{y}}^{(r)})=1$. The above matrix is $\begin{pmatrix}\mathbf{A}_{1,1}&\mathbf{0}_{(3N-m)\times m}\end{pmatrix}^{T}$. 
Through column transformations, the matrix $\mathbf{A}_{1,1}$ can be transformed to
\begin{equation}
    \begin{pmatrix}
        \chi^{[m]}&S_{1}(\hat{\mathbf{y}}^{(1)})&S_{2}(\hat{\mathbf{y}}^{(1)})(\chi^{[m]})^{2}&\cdots &S_{m-1}(\hat{\mathbf{y}}^{(1)})\chi^{[m]}\\
0&1&S_{1}(\hat{\mathbf{y}}^{(2)})(\chi^{[m]})^{2}&\cdots &S_{m-2}(\hat{\mathbf{y}}^{(2)})\chi^{[m]}\\
0&0&(\chi^{[m]})^{2}&\cdots &S_{m-3}(\hat{\mathbf{y}}^{(3)})\chi^{[m]}\\
\vdots&\vdots&\vdots&\ddots &S_{1}(\hat{\mathbf{y}}^{(m-1)})\chi^{[m]}\\
0&0&0&\cdots &\chi^{[m]}\\
    \end{pmatrix}.
\end{equation}
Moreover, the $m$-th order principal minor of $\mathbf{Y}_{N}^{(2)}(\mathbf{x})$ is $\det(\mathbf{A}_{1,1})=(\chi^{[m]})^{m}$, which is independent of the parameter.

Now we look at the other columns. Expand
\begin{equation}
    S_{n+3lm}(\hat{\mathbf{x}}^{(r)})=\sum_{i=[n/m]+2l}^{[n/m]+3l}\frac{(\chi^{[m]})^{i}}{i!}S_{n-(i-3l)m}(\hat{\mathbf{y}}^{(r)})
    +\sum_{i=0}^{[n/m]+2l-1}\frac{(\chi^{[m]})^{i}}{i!}S_{n-(i-3l)m}(\hat{\mathbf{y}}^{(r)})
\end{equation}
and denote $s_{i}=\frac{(\chi^{[m]})^{[n/m]+2i}}{([n/m]+2i)!}S_{n-([n/m]-i)m}(\hat{\mathbf{y}}^{(r)})$, it follows that
\begin{equation}
    \begin{split}
        S_{n+3lm}(\hat{\mathbf{x}}^{(r)})=&\sum_{i=0}^{l}\frac{(\chi^{[m]})^{i+[n/m]+2l}}{(i+[n/m]+2l)!}S_{n-([n/m]-l+i)m}(\hat{\mathbf{y}}^{(r)})+\mathcal{O}((a^{[m]})^{[n/m]+2l-1})\\
        =&\sum_{i=0}^{l}\frac{(\chi^{[m]})^{-i+[n/m]+3l}}{(-i+[n/m]+3l)!}S_{n-([n/m]-i)m}(\hat{\mathbf{y}}^{(r)})+\mathcal{O}((a^{[m]})^{[n/m]+2l-1})\\
        =&\sum_{i=0}^{l}\frac{([n/m]+2i)!}{(-i+[n/m]+3l)!}(\chi^{[m]})^{3(l-i)}s_{i}+\mathcal{O}((a^{[m]})^{[n/m]+2l-1}).
        \end{split}
\end{equation}
For fixed $l$, we can use $S_{n+3im},0\leq i\leq l-1$ to remove the terms $(\chi^{[m]})^{i},2l+1\leq i \leq 3l$ in $S_{n+3lm}$. 
Moreover, for $n,n-1,\cdots,[\frac{n}{m}]m$, the coefficients of $s_{i}$ are just $\frac{([n/m]+2i)!}{(-i+[n/m]+3l)!}(\chi^{[m]})^{3(l-i)}$. 
Since $\{s_{i},1\leq i \leq l\}$ are linearly independent, and the determinant of column transformation matrices are constant 
multiple of the power of $\chi^{[m]}$, 
we can only preserve the coefficients of $(\chi^{[m]})^{[n/m]+2l}$ on $j$-th $m$ columns, i.e.
$\frac{([n/m])!}{([n/m]+3l)!}S_{n-([n/m]-l)m}(\hat{\mathbf{y}}^{(r)})$. For example, we calculate $\mathbf{A}_{1,2}$ and $\mathbf{A}_{2,2}$. 
If $m=3j+1,j\geq 1$, then we reserve the term $(\chi^{[m]})^{2},(\chi^{[m]})^{3}$ for $m+1$ to $m+j$ columns, 
with respect to $(\chi^{[m]})^{3},(\chi^{[m]})^{4}$ for $m+j+1$ to $m+2j+1$ columns and 
with respect to $(\chi^{[m]})^{4},(\chi^{[m]})^{5}$ for $m+2j+2$ to $m+3j+1=2m$ columns. Then $m+1$ to $2m$ columns of $\mathbf{Y}_{N}^{(2)}(\mathbf{x})$ 
can be approximated by $3N\times m$ matrix
\begin{equation}\label{inner-exp}
    \begin{split}
        \begin{pmatrix}
            S_{m+1}(\hat{\mathbf{y}}^{(1)})(\chi^{[m]})^{2} & S_{m}(\hat{\mathbf{y}}^{(2)})(\chi^{[m]})^{2} & S_{m-1}(\hat{\mathbf{y}}^{(3)})(\chi^{[m]})^{2} & \cdots \\
            S_{m+4}(\hat{\mathbf{y}}^{(1)})(\chi^{[m]})^{2} & S_{m+3}(\hat{\mathbf{y}}^{(2)})(\chi^{[m]})^{2} & S_{m+2}(\hat{\mathbf{y}}^{(3)})(\chi^{[m]})^{2} & \cdots \\
            \vdots & \vdots & \vdots & \ddots  \\
            S_{m+3j-2}(\hat{\mathbf{y}}^{(1)})(\chi^{[m]})^{2} & S_{m+3j-3}(\hat{\mathbf{y}}^{(2)})(\chi^{[m]})^{2} & S_{m+3j-4}(\hat{\mathbf{y}}^{(3)})(\chi^{[m]})^{2} & \cdots \\
            S_{m}(\hat{\mathbf{y}}^{(1)})(\chi^{[m]})^{3} & S_{m-1}(\hat{\mathbf{y}}^{(2)})(\chi^{[m]})^{3} & S_{m-2}(\hat{\mathbf{y}}^{(3)})(\chi^{[m]})^{3} & \cdots \\
            \vdots & \vdots & \vdots & \ddots \\
            S_{m+3j}(\hat{\mathbf{y}}^{(1)})(\chi^{[m]})^{3} & S_{m+3j-1}(\hat{\mathbf{y}}^{(2)})(\chi^{[m]})^{3} & S_{m+3j-2}(\hat{\mathbf{y}}^{(3)})(\chi^{[m]})^{3} & \cdots \\
            S_{m+2}(\hat{\mathbf{y}}^{(1)})(\chi^{[m]})^{4} & S_{m+1}(\hat{\mathbf{y}}^{(2)})(\chi^{[m]})^{4} & S_{m}(\hat{\mathbf{y}}^{(3)})(\chi^{[m]})^{4} & \cdots \\
            \vdots & \vdots & \vdots & \ddots \\
            S_{m+3j-1}(\hat{\mathbf{y}}^{(1)})(\chi^{[m]})^{4} & S_{m+3j-2}(\hat{\mathbf{y}}^{(2)})(\chi^{[m]})^{4} & S_{m+3j-3}(\hat{\mathbf{y}}^{(3)})(\chi^{[m]})^{4} & \cdots \\
        \end{pmatrix}^{T}+
        \begin{pmatrix}
            (\chi^{[m]})^{3}\mathbf{A}_{1,1} \\
            \mathbf{0}_{(3N-m)\times m}
        \end{pmatrix}.
    \end{split}
\end{equation}
Define the first matrix of \eqref{inner-exp} be $\begin{pmatrix}\mathbf{A}_{1,2}& \mathbf{A}_{2,2} &\mathbf{0}_{(3N-2m)\times m}\end{pmatrix}^{T}$. 
Then $(\chi^{[m]})^{-2}\mathbf{A}_{2,2}$ is
\begin{equation*}
        \begin{pmatrix}
            S_{1}(\hat{\mathbf{y}}^{(m+1)})  & \cdots & S_{3j-2}(\hat{\mathbf{y}}^{(m+1)}) & \chi^{[m]} & \cdots & S_{3j}(\hat{\mathbf{y}}^{(m+1)})\chi^{[m]} & S_{2}(\hat{\mathbf{y}}^{(m+1)})(\chi^{[m]})^{2} & \cdots \\
            1  & \cdots & S_{3j-3}(\hat{\mathbf{y}}^{(m+2)}) & 0 & \cdots & S_{3j-1}(\hat{\mathbf{y}}^{(m+2)})\chi^{[m]} & S_{1}(\hat{\mathbf{y}}^{(m+2)})(\chi^{[m]})^{2} & \cdots \\
            0  & \cdots & S_{3j-4}(\hat{\mathbf{y}}^{(m+3)}) & 0 & \cdots & S_{3j-2}(\hat{\mathbf{y}}^{(m+3)})\chi^{[m]} & (\chi^{[m]})^{2} & \cdots \\
            \vdots & \ddots & \vdots & \vdots & \ddots & \vdots & \vdots & \ddots 
    \end{pmatrix}
\end{equation*}
which has similar form as $\mathbf{A}_{1,1}$. 
Repeating this process, then we obtain all $\mathbf{A}_{i,j}$. Moreover, 
all $\mathbf{A}_{i,i}$ can be transformed to upper triangular matrices through 
column transformations. The determinants $\det(\mathbf{A}_{i,i})=(\chi^{[m]})^{2m(i-1)}\det(\mathbf{A}_{1,1})=(\chi^{[m]})^{2mi-m}$. 

Now we need to calculate $\mathbf{C}_{k+1,k+1}$ in \eqref{Yasym}. 
Through the above analysis, we can only preserve the coefficient of
$(\chi^{[m]})^{[n/m]+2l}$ which just $\frac{([n/m])!}{([n/m]+3k)!}S_{n-([n/m]-k)m}(\hat{\mathbf{y}}^{(r)})$ for 
the last $N_{0}$ columns of $\mathbf{Y}_{N}^{(2)}(\mathbf{x})$. 
Notice that $S_{n-([n/m]-k)m}(\hat{\mathbf{y}}^{(r)})=S_{n-[n/m]m+km}(\hat{\mathbf{y}}^{(r)})$, and $n-[n/m]m$ is the remainder of $n$ mod $m$.
Since the subscript of $S_{i}$ minus one if the column index minus one, when the column index reduces $km$, 
the coefficient is $S_{n-[n/m]m}(\hat{\mathbf{y}}^{(r)})$. Hence $\mathbf{C}_{k+1,k+1}$ have similar form of $\mathbf{A}_{1,1}$. 
We just need to concentrate on the first $m$ columns to calculate
the last $N_{0}$ columns. Let $m=3j+1,j\geq 1$, then $\mathbf{C}_{k+1,k+1}$ is formed by the first $2N+N_{0}$ rows and 
$N_{0}$ columns of the matrix
\begin{equation}\label{remain-5}
    \begin{pmatrix}
        S_{1}(\hat{\mathbf{y}}^{(km+1)})(\chi^{[m]})^{2k} & (\chi^{[m]})^{2k} & 0 & \cdots \\
        \vdots & \vdots & \vdots & \ddots \\
        S_{3j-2}(\hat{\mathbf{y}}^{(km+1)})(\chi^{[m]})^{2k} & S_{3j-3}(\hat{\mathbf{y}}^{(km+2)})(\chi^{[m]})^{2k} & S_{3j-4}(\hat{\mathbf{y}}^{(km+3)})(\chi^{[m]})^{2k} & \cdots \\
        (\chi^{[m]})^{2k+1}& 0 & 0 & \cdots \\
        \vdots & \vdots & \vdots & \ddots \\
        S_{3j}(\hat{\mathbf{y}}^{(km+1)})(\chi^{[m]})^{2k+1}& S_{3j-1}(\hat{\mathbf{y}}^{(km+2)})(\chi^{[m]})^{2k+1}& S_{3j-2}(\hat{\mathbf{y}}^{(km+3)})(\chi^{[m]})^{2k+1}&\cdots \\
        S_{2}(\hat{\mathbf{y}}^{(km+1)})(\chi^{[m]})^{2k+2}&S_{1}(\hat{\mathbf{y}}^{(km+2)})(\chi^{[m]})^{2k+2}&(\chi^{[m]})^{2k+2}&\cdots \\
        \vdots & \vdots & \vdots & \ddots \\
        S_{3j-1}(\hat{\mathbf{y}}^{(km+1)})(\chi^{[m]})^{2k+2}&S_{3j-2}(\hat{\mathbf{y}}^{(km+2)})(\chi^{[m]})^{2k+2}&S_{3j-3}(\hat{\mathbf{y}}^{(km+3)})(\chi^{[m]})^{2k+2}&\cdots 
    \end{pmatrix}_{m \times (2N+m)}^{T}.
\end{equation}
To simplify the notations, we omit the variables and the power of $\chi^{[m]}$ in the following. 
For example, the first column of the matrix \eqref{remain-5} is of the form
\begin{equation}
    \begin{pmatrix}
        S_{1}&S_{4}&\cdots&S_{3j-2}&S_{0}&\cdots&S_{3j}&S_{2}&\cdots&S_{3j-1}
    \end{pmatrix}_{1\times m}.
\end{equation}

If $1\leq N_{0} \leq j$, the last $N_{0}$ columns of the first row is of the form $(S_{1},S_{4},\cdots,S_{3N_{0}-2})_{1\times N_{0}}$ which can be expressed 
by $N_{0}$th order 1-type rogue wave solutions. 

If $j+1\leq N_{0} \leq 2j+1$, the last $N_{0}$ columns of the first row is $(S_{1},S_{4},\cdots,S_{3j-2},S_{0},\cdots,S_{3(N_{0}-j-1)})_{1\times N_{0}}$. 
Since there is only one nonzero element in the $(j+1)$th column and there are two nonzero elements in the first column 
of matrix $\mathbf{C}_{k+1,k+1}$, when calculating the term 
$\det((\mathbf{C}_{k+1,k+1}(\mathbf{x}^{(2)}))^{\dagger}\mathbf{C}_{k+1,k+1}(\mathbf{x}^{(1)}))$ in \eqref{asy-inner}, 
we can discard these two columns:
\begin{equation}\label{discard1}
    \begin{split}
    \begin{pmatrix}
        S_{1}&S_{4}&\cdots&S_{3j-2}&S_{0}&S_{3}&\cdots&S_{3(N_{0}-j-1)}\\
        S_{0}&S_{3}&\cdots&S_{3j-3}&0&S_{2}&\cdots&S_{3(N_{0}-j-2)}\\
        0&S_{2}&\cdots&S_{3j-4}&0&S_{1}&\cdots&S_{3(N_{0}-j-3)}\\
        \vdots&\vdots&\ddots&\vdots&\vdots&\vdots&\ddots&\vdots
    \end{pmatrix}
    \to
    \begin{pmatrix}
        S_{2}&\cdots&S_{3j-4}&S_{1}&\cdots&S_{3(N_{0}-j-3)}\\
        \vdots&\ddots&\vdots&\vdots&\ddots&\vdots
    \end{pmatrix}.
    \end{split}
\end{equation}
Denote $\mathbf{D}_{1}$ be the submatrix of matrix $\mathbf{C}_{k+1,k+1}$ obtained 
by removing the first and second rows, the first column, and the $(j+1)$th column, 
we have approximation
\begin{equation}
    \det((\mathbf{C}_{k+1,k+1}(\mathbf{x}^{(2)}))^{\dagger}\mathbf{C}_{k+1,k+1}(\mathbf{x}^{(1)}))\sim
    \det((\mathbf{D}_{1}(\mathbf{x}^{(2)}))^{\dagger}\mathbf{D}_{1}(\mathbf{x}^{(1)}))
\end{equation}
in \eqref{asy-inner}. 
The matrix $\mathbf{D}_{1}$ is a $(2N+N_{0}-2)$ by $(j-1)+(N_{0}-j-1)=(N_{0}-2)$ matrix. 
Hence the remaining matrix $\mathbf{D}_{1}$ corresponds to $(j-1)$th order $0$-type and $(N_{0}-j-1)$th 
order $1$-type rogue wave solutions. 

If $2j+2\leq N_{0} \leq 3j$, for the matrix $\mathbf{C}_{k+1,k+1}$, we observe that the columns $j+1, 1, 2j+2$ contain one, two, and three nonzero 
elements respectively.  Furthermore, the columns $j+2, 2, 2j+3$ contain four, five, and six nonzero elements respectively. 
Based on this observation, we can continue this process. 
Using the same method in \eqref{discard1}, 
if we denote $\mathbf{D}_{2}$ be the submatrix of 
matrix $\mathbf{C}_{k+1,k+1}$ obtained by removing the first $3(N_{0}-2j-1)$ rows, 
the first column to the $(N_{0}-2j-1)$th column, the $(j+1)$th column to the $(N_{0}-j-1)$th column, 
and the $(2j+2)$th column to the $N_{0}$th column, then we have approximation
\begin{equation}
    \det((\mathbf{C}_{k+1,k+1}(\mathbf{x}^{(2)}))^{\dagger}\mathbf{C}_{k+1,k+1}(\mathbf{x}^{(1)}))\sim
    \det((\mathbf{D}_{2}(\mathbf{x}^{(2)}))^{\dagger}\mathbf{D}_{2}(\mathbf{x}^{(1)}))
\end{equation}
in \eqref{asy-inner}. Hence the remaining matrix $\mathbf{D}_{2}$ corresponds to 
$(m-1-N_{0})$th order $0$-type and $(m-N_{0})$th order $1$-type rogue wave solutions. 

For the case $m=3j+2$, we can use the same method. The conclusion is given in the proposition. 

\end{proof}
The rogue wave patterns in the inner region reveal that, at points around the origin, 
the rogue wave solutions can be approached by lower-order rogue waves. 
Now we can see why $\chi^{[3i]}, i\geq 1$ (the case $m$ is a multiple of $3$) would not affect the rogue wave solutions. 
Using the expansion \eqref{expand1}, the term $\mathbf{Y}_{N}^{(2)}(\mathbf{x})$ would be a
$3N\times m$ matrix
\begin{equation}
    \begin{pmatrix}
        S_{1}(\hat{\mathbf{y}}^{(1)})  & \cdots & S_{m-2}(\hat{\mathbf{y}}^{(1)}) & S_{m+1}(\hat{\mathbf{y}}^{(1)})+S_{1}(\hat{\mathbf{y}}^{(1)})\chi^{[m]} & \cdots  \\
        S_{0}(\hat{\mathbf{y}}^{(2)})  & \cdots & S_{m-3}(\hat{\mathbf{y}}^{(2)}) & S_{m}(\hat{\mathbf{y}}^{(2)})+S_{0}(\hat{\mathbf{y}}^{(1)})\chi^{[m]} & \cdots  \\
        \vdots & \ddots & \vdots & \vdots & \ddots 
    \end{pmatrix}.
\end{equation}
We can eliminate the term with the power of $\chi^{[m]}$ in the $(m/3+1)$-th row using the first row, and repeat this process. 
Then the rogue wave solutions \eqref{special-rw2} would not contain the term $\chi^{[3i]}, i\geq 1$. 

Combining the rogue wave patterns in the inner and outer regions, we obtain the asymptotic expression of rogue waves generated at 
the branch point of multiplicity three on the Riemann surfaces. And the distribution of these 
first-order rogue waves can be represented 
by the roots of Okamoto polynomial hierarchies and the center \eqref{center}: 
\begin{thm}[Rogue wave patterns]\label{decopose}
    Let $\eta=(\chi^{[m]})^{1/m}, m\geq 2$ and $m$ is not a multiple of $3$, 
    where $\chi^{[m]}$ is an internal parameter of the $k$-type rogue wave solutions 
    \eqref{special-rw2}, suppose the nonzero roots of Okamoto polynomial 
    hierarchies $W_{N}^{[k,m]}(z)$ 
    are all simple. As $|\eta|\to \infty$, we can decompose the rogue wave solutions \eqref{special-rw2} 
    into $N(N+1-k)-N^{[k]}$ first-order rogue wave solutions \eqref{outer-solution} in the outer region 
    and a lower order rogue wave solutions \eqref{inner-solution} in the inner region
    \begin{equation}\label{deco}
        u_{i,k}^{[N]}(x,t)=\sum_{(x_0,t_0)} u_{i,k}^{asy}(x-x_{0}|\eta|,t-t_{0}|\eta|)+u_{i,k}^{(N_{1}^{[k]},N_{2}^{[k]})}(x,t)
        +\mathcal{O}(\eta^{-1}),
    \end{equation}
    where $(x_{0},t_{0})$ traverses the nonzero roots of $W_{N}^{[k,m]}((\vartheta^{[1]}(x,t)-\chi^{[1]}){\rm e}^{-{\rm i}\arg{\eta}})$. The positions of these 
    first-order rogue waves in the outer region are $(x_{0}|\eta|+\hat{x}_{1},t_{0}|\eta|+\hat{t}_{1})$, where 
    $(\hat{x}_{1},\hat{t}_{1})$ is defined in \eqref{center}. The position of the lower rogue wave in the inner region is 
    the origin. 
\end{thm}
In conclusion, Theorem \ref{decopose} tells us the decomposition of rogue wave solutions \eqref{special-rw2} 
when one of the internal parameters is large enough. 
The positions and the orders of rogue waves in the outer region correspond to the roots of Okamoto polynomial hierarchies. 
Due to the root distributions of the Okamoto polynomial hierarchies given by \eqref{Ok-poly}, 
we can observe the symmetry structures by the positions of the rogue wave patterns.

\subsection{Examples}
Now we give some examples to verify and visualize the rogue wave patterns in Theorem \ref{decopose}. 
Assuming $b_{1}=1$, $b_{2}=2$, and $N=3$, we can calculate 
\begin{equation*}
    \vartheta^{[1]}(x,t)=-\frac{2{\rm i}}{3}x+\frac{\sqrt{3}+{\rm i}}{9}t+\chi^{[1]}. 
\end{equation*}
To simplify the calculation, without loss of generality, 
we consider the parameter settings where all internal parameters are zero except for one that is nonzero. 
Using this assumption, we plot the graph of the norm of rogue wave solutions \eqref{special-rw2} and the 
positions which are given in Theorem \ref{decopose}. 

\begin{figure}[htbp]
    \centering
    \includegraphics[scale=0.205]{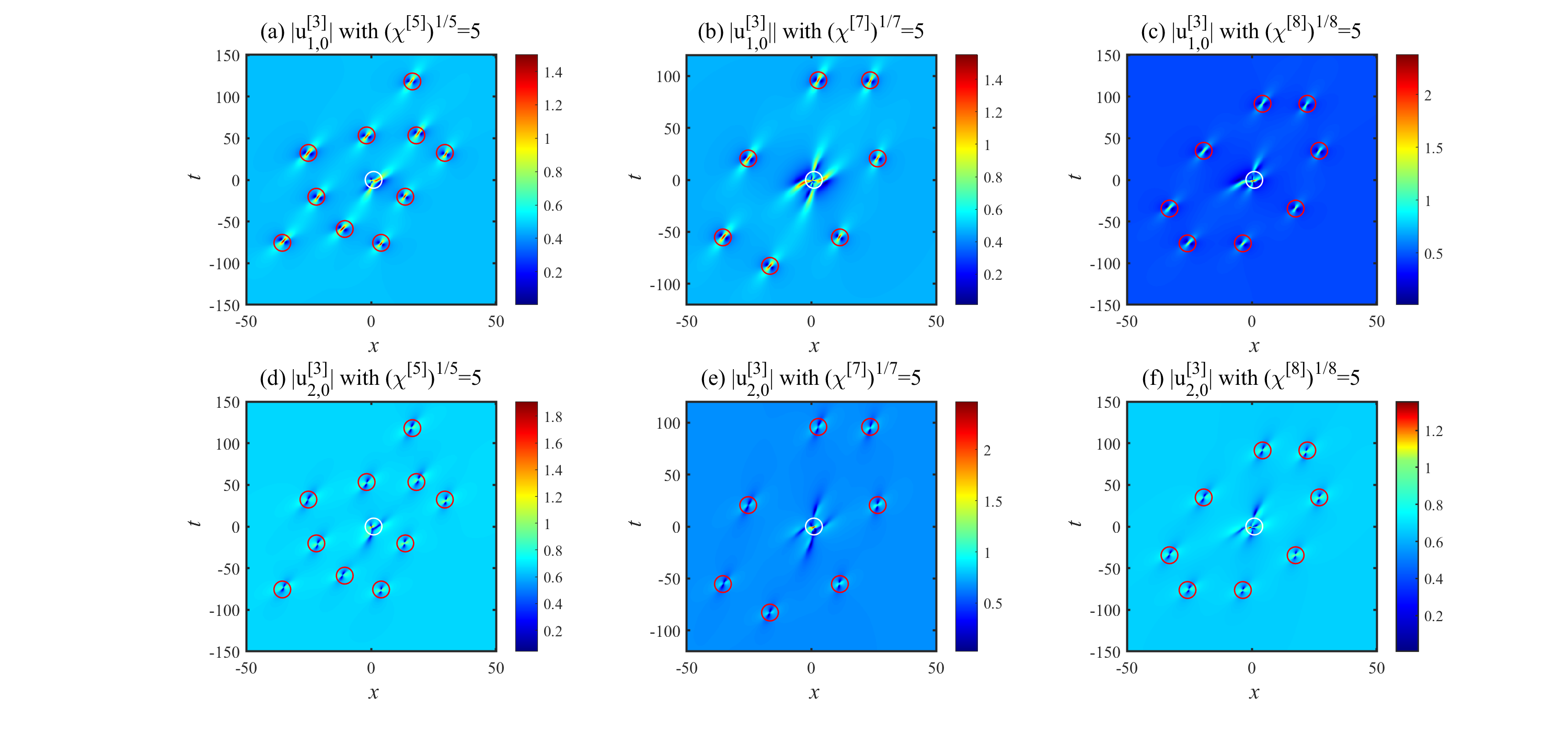}
    \caption{The predicted positions of $u_{i,0}^{[3]}(x,t)$ in Theorem \ref{decopose}. 
    Parameter settings: (a,d) $(\chi^{[5]})^{1/5}=5$, (b,e) $(\chi^{[7]})^{1/5}=5$ and (c,f) $(\chi^{[8]})^{1/5}=5$. 
    First row: $|u_{1,0}^{[3]}(x,t)|$. Second row: $|u_{2,0}^{[3]}(x,t)|$. 
    The backgrounds are the norm of 0-type rogue wave solutions $u_{i,0}^{[3]}(x,t)$. 
    The red circles are the predicted positions 
    of first-order rogue wave solutions in Theorem \ref{decopose} in the outer region, and
    the white circles are the predicted positions of the lower rogue wave solutions 
    in the inner region.}
    \label{N30}
\end{figure}
\begin{enumerate}
    \item For 0-type rogue waves, we consider three cases
    \begin{equation*}
        (\chi^{[5]})^{1/5}=5,\quad (\chi^{[7]})^{1/7}=5,\quad (\chi^{[8]})^{1/8}=5, 
    \end{equation*}
    and the figures are plotted in Figure \ref{N30}. In these cases, the degree of $W_{3}^{[0,m]}(z)$ with respect 
    to $z$ is $N(N+1)=12$. 
    \begin{enumerate}
        \item For $(\chi^{[5]})^{1/5}=5$, the first component of the solution \eqref{special-rw2} is plotted in Figure (\ref{N30}-a) and 
        the second component is plotted in Figure (\ref{N30}-d). Since $m=5,N=3$, we calculate $(N_{1}^{[0]},N_{2}^{[0]})=(1,1)$ and $N^{[0]}=2$
        in Theorem \ref{Ok-root}. Hence there are $N(N+1)-N^{[0]}=10$ first-order rogue wave solutions in the outer region and 
        a $(1,1)$-order rogue wave solution in the inner region. 
        \item For $(\chi^{[7]})^{1/7}=5$, the two components are plotted in Figure (\ref{N30}-b,e). The term 
        $(N_{1}^{[0]},N_{2}^{[0]})=(2,1)$ and $N^{[0]}=5$. Hence there are $N(N+1)-N^{[0]}=7$ first-order rogue wave solutions in 
        the outer region and a $(2,1)$-order rogue wave solution in the inner region. 
        \item For $(\chi^{[8]})^{1/8}=5$, the two components are plotted in Figure (\ref{N30}-c,f). By direct calculation, 
        there are $8$ first-order rogue wave solutions in the outer region and a $(0,2)$-order rogue wave solution in the inner region, 
        i.e. a second-order 1-type rogue wave solution. 
    \end{enumerate}
    \item For 1-type rogue waves, we consider four cases
    \begin{equation*}
        (\chi^{[2]})^{1/2}=5,\quad (\chi^{[4]})^{1/4}=5,\quad (\chi^{[5]})^{1/5}=5,\quad (\chi^{[7]})^{1/7}=5,
    \end{equation*}
    and the figures are plotted in Figure \ref{N31}. The degree of $W_{3}^{[1,m]}(z)$ with respect 
    to $z$ is $N^{2}=9$. 
    \begin{enumerate}
        \item For $(\chi^{[2]})^{1/2}=5$, the first component of the solution \eqref{special-rw2} is plotted in Figure (\ref{N31}-a) and 
        the second component is plotted in Figure (\ref{N31}-e). Since $m=2$, the terms $(N_{1}^{[1]},N_{2}^{[1]})=(0,1)$ and $N^{[1]}=1$
        in Theorem \ref{Ok-root}. There are $8$ first-order rogue wave solutions in the outer region and 
        a $(0,1)$-order rogue wave solution in the inner region, i.e. a first-order 1-type rogue wave solution. 
        \item For $(\chi^{[4]})^{1/4}=5$, the two components are plotted in Figure (\ref{N31}-b,f). The 
        order $(N_{1}^{[1]},N_{2}^{[1]})=(0,1)$ and $N^{[1]}=1$ in Theorem \ref{Ok-root}. 
        The case $(\chi^{[4]})^{1/4}=5$ is the same as $(\chi^{[2]})^{1/2}=5$. 
        \item For $(\chi^{[5]})^{1/5}=5$, the two components are plotted in Figure (\ref{N31}-c,g). We get 
        $(N_{1}^{[1]},N_{2}^{[1]})=(2,1)$ and $N^{[1]}=4$. Hence there are $5$ first-order rogue wave solutions in 
        the outer region and a $(2,1)$-order rogue wave solution in the inner region. 
        \item For $(\chi^{[5]})^{1/5}=5$, the two components are plotted in Figure (\ref{N31}-d,h). By straightforward calculation, 
        there are $7$ first-order rogue wave solutions in the outer region and a $(1,0)$-order rogue wave solution in the inner region, 
        i.e. a first-order 0-type rogue wave solution. 
    \end{enumerate}
\end{enumerate}

\begin{figure}[htbp]
    \centering
    \includegraphics[scale=0.205]{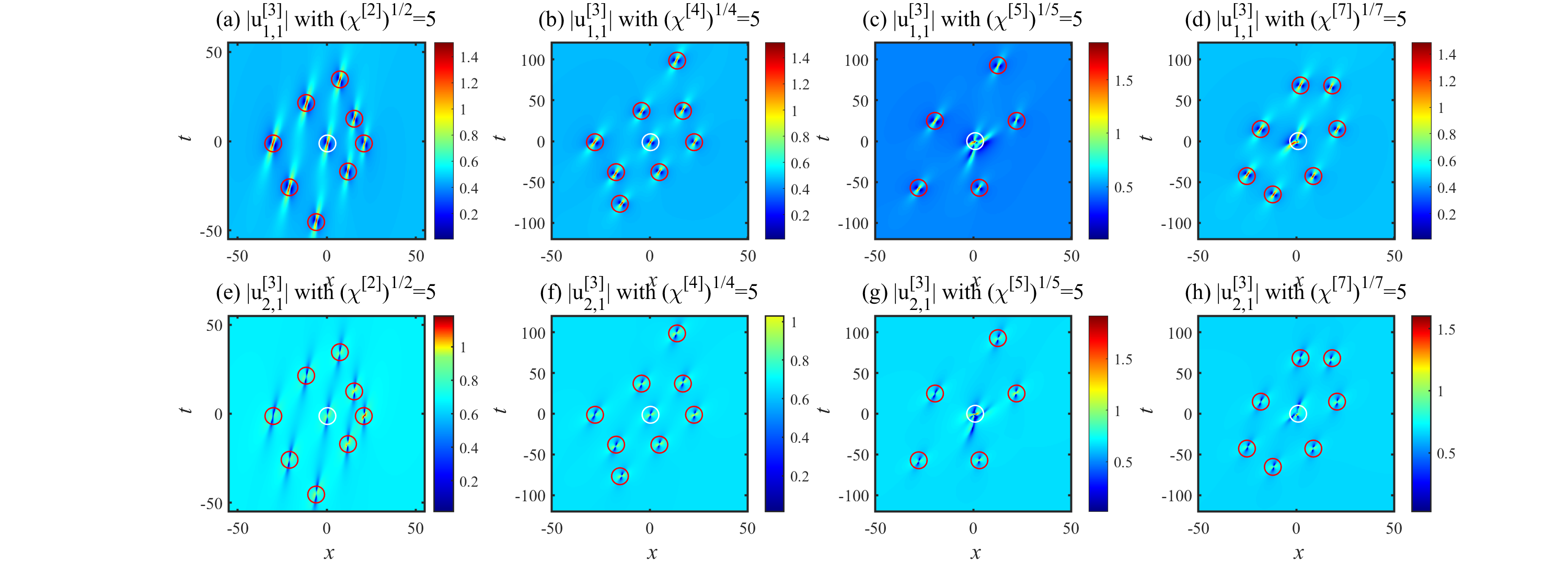}
    \caption{The predicted positions of $u_{i,1}^{[3]}(x,t)$ in Theorem \ref{decopose}. Parameter settings: 
    (a,e) $(\chi^{[2]})^{1/2}=5$, (b,f) $(\chi^{[4]})^{1/4}=5$, (c,g) $(\chi^{[5]})^{1/5}=5$ and (d,h) $(\chi^{[7]})^{1/7}=5$.
    First row: $|u_{1,1}^{[3]}(x,t)|$. Second row: $|u_{2,1}^{[3]}(x,t)|$. The backgrounds are the norm of $u_{i,1}^{[3]}(x,t)$. 
    The circles are the predicted positions of the rogue waves in the outer and inner regions in Theorem \ref{decopose}.}
    \label{N31}
\end{figure}


From Figure \ref{N30} and \ref{N31}, it can be observed that these circles predict the positions of the rogue wave solutions 
in Theorem \ref{decopose}. 
We can use the roots of Okamoto polynomial hierarchies 
$W_{N}^{[k,m]}(\vartheta^{[1]}(x,t)-\chi^{[1]})$ and centers \eqref{center} to predict the positions. 

Now we adjust the argument of $(\chi^{[m]})^{1/m}$ with fixed norm. 
Denote $\tilde{\theta}=\arg{\left((\chi^{[m]})^{1/m}\right)}$, let 
$(x_{\tilde{\theta}},t_{\tilde{\theta}})$ be the roots of $W_{N}^{[k,m]}((\vartheta^{[1]}(x,t)-\chi^{[1]}){\rm e}^{-{\rm i}\tilde{\theta}})$ 
for $\tilde{\theta} \in [0,2\pi)$. 
Using the expansion of $\vartheta$ in \eqref{exp-theta}, we define the matrix $\mathbf{A}$ which is given by
\begin{equation*}
    \begin{pmatrix}
        \vartheta^{[1]}-\chi^{[1]}\\
        (\vartheta^{[1]}-\chi^{[1]})^{*}
    \end{pmatrix}
    =\mathbf{A}
    \begin{pmatrix}
        x\\
        t
    \end{pmatrix}.
\end{equation*}
It leads to the relation
\begin{equation}\label{co-tran}
    \begin{pmatrix}
        x_{\tilde{\theta}}\\
        t_{\tilde{\theta}}
    \end{pmatrix}
    =\mathbf{A}^{-1}
    \mathrm{diag}\left({\rm e}^{{\rm i}\tilde{\theta}},{\rm e}^{-{\rm i}\tilde{\theta}}\right)
    \mathbf{A}
    \begin{pmatrix}
        x_{\tilde{\theta}=0}\\
        t_{\tilde{\theta}=0}
    \end{pmatrix}.
\end{equation}
Note that the matrix $\mathbf{A}^{-1}\mathrm{diag}\left({\rm e}^{{\rm i}\tilde{\theta}},{\rm e}^{-{\rm i}\tilde{\theta}}\right)\mathbf{A}$ has real elements, 
since $\mathbf{A}$ is the transformation matrix between $(x,t)^{T}$ and a pair of conjugate complex numbers. 
Hence for the parameter $(\chi^{[m]})^{1/m}$ with fixed norm, 
we can use the position of case $\tilde{\theta}=0$ to calculate the position $(x_{\tilde{\theta}},t_{\tilde{\theta}})$ 
for $\tilde{\theta} \in [0,2\pi)$ through a coordinate transformation \eqref{co-tran}. 
Now we give an example for the cases different $\tilde{\theta}$, with fixed norm $|(\chi^{[m]})^{1/m}|$. 
\begin{figure}[htbp]
    \centering
    \includegraphics[scale=0.205]{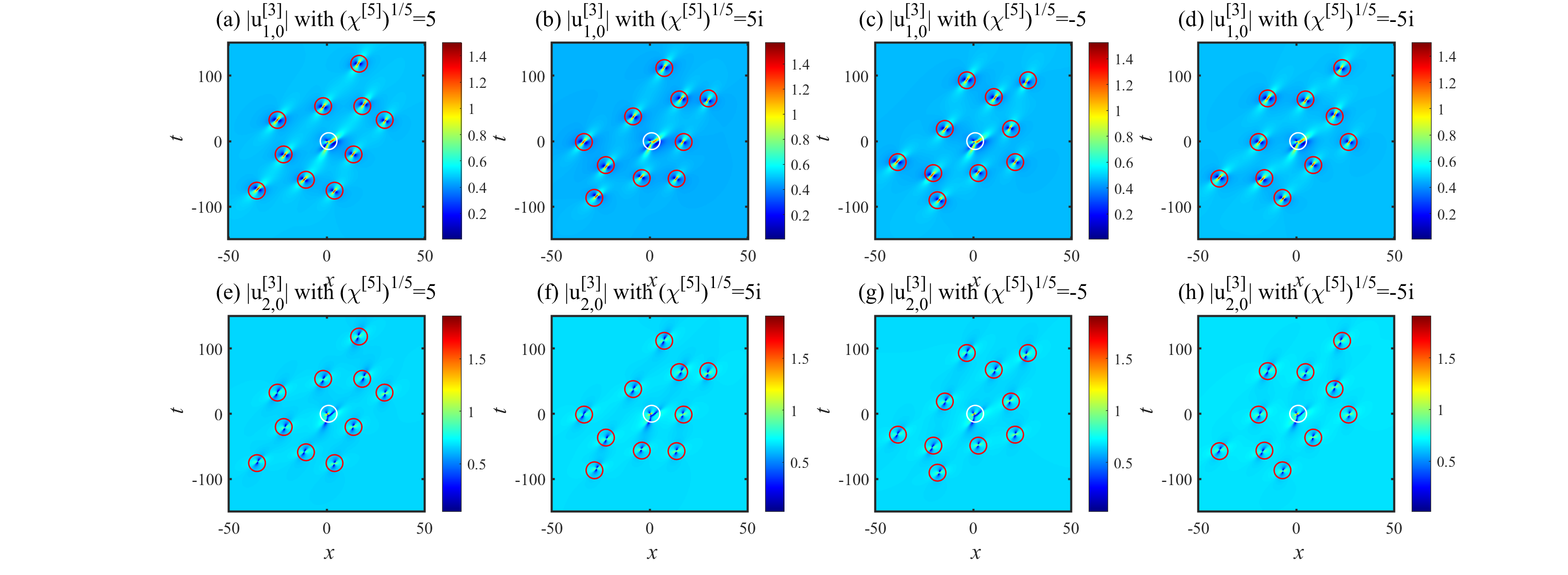}
    \caption{The predicted positions of $u_{i,0}^{[3]}(x,t)$ with rotation. Parameter settings: 
    (a,e) $(\chi^{[5]})^{1/5}=5$, (b,f) $(\chi^{[5]})^{1/5}=5{\rm i}$, (c,g) $(\chi^{[5]})^{1/5}=-5$ and (d,h) 
    $(\chi^{[5]})^{1/5}=-5{\rm i}$. First row: $|u_{1,0}^{[3]}(x,t)|$. 
    Second row: $|u_{2,0}^{[3]}(x,t)|$. The backgrounds are the norm of $u_{i,0}^{[3]}(x,t)$. The circles are 
    the predicted positions in Theorem \ref{decopose}. } 
    \label{N30arg}
\end{figure}
We set $b_{1}=1,b_{2}=2,N=3$ and the parameters $\chi^{[k]}=0,k\ne 5$ with four cases
\begin{equation}
    (\chi^{[5]})^{1/5}=5{\rm e}^{{\rm i}\tilde{\theta}},\quad \tilde{\theta}=0,\frac{\pi}{2},\pi,\frac{3\pi}{2}. 
\end{equation}
Here we consider 0-type rogue waves, and plot the norm of rogue wave solutions $|u_{i,0}^{[3]}(x,t)|$ and the positions in Figure \ref{N30arg}. 
As we change $\tilde{\theta}$, the rogue wave patterns have rotations, which are given by the formula \eqref{co-tran}. 

\section{Conclusion}\label{sec-col}
In this paper, we use the Darboux transformation to construct the rogue wave solutions \eqref{gene-rw2}
of the CFL equations \eqref{CFL} based on the paper \cite{ling_general_2018}. 
With the aid of the form \eqref{special-rw2} of $k$-type rogue wave solutions, we can analyze the rogue wave patterns for \eqref{special-rw2}. 
The patterns of the rogue wave solutions generated at the branch point of multiplicity three 
are determined by the root structures of the 
Okamoto polynomial hierarchies with a linear transformation based on the paper \cite{yang_rogue_2022}. 
After letting one of the internal parameters large enough, 
the Okamoto polynomial hierarchies \eqref{Okamoto} arise naturally and the rogue wave solutions \eqref{special-rw2} have 
decomposition \eqref{deco}. 
We can predict the positions of the first-order rogue waves in \eqref{deco} using the root distributions \eqref{Ok-poly}
of the Okamoto polynomial hierarchies. 

For the CFL equations, we can also consider the rogue wave solutions generated at the branch point of multiplicity two 
and the rogue wave patterns are associated with Yablonskii-Vorob’ev polynomial hierarchies. 
More generally, in other models of integrable systems, we can use the roots of special polynomials to study 
the patterns of the rogue wave solutions with tau function determinant representations. 
Specifically, for general integrable models, we can also construct its Darboux transformation 
and use the seed solutions to generate the high-order rogue wave solutions at the branch point of the Riemann surface. 
For such rogue wave solutions, we can use a similar approach to calculate the asymptotic expressions and 
analyze the properties of the associated polynomial hierarchies to study the rogue wave patterns.

\bibliographystyle{siam}
\bibliography{CFL-pattern}

\begin{thebibliography}{10}

\bibitem{baronio_observation_2018}
{\sc F.~Baronio, B.~Frisquet, S.~Chen, G.~Millot, S.~Wabnitz, and B.~Kibler},
  {\em Observation of a group of dark rogue waves in a telecommunication
  optical fiber}, Phys. Rev. A, 97 (2018), p.~013852.

\bibitem{boyd_nonlinear_2020}
{\sc R.~W. Boyd}, {\em Nonlinear Optics}, Academic press, 2020.

\bibitem{chen_integrability_1979}
{\sc H.~H. Chen, Y.~C. Lee, and C.~S. Liu}, {\em Integrability of nonlinear
  {Hamiltonian} systems by inverse scattering method}, Phys. Scr., 20 (1979),
  pp.~490--492.

\bibitem{chen_peregrine_2014}
{\sc S.~Chen and L.-Y. Song}, {\em Peregrine solitons and algebraic soliton
  pairs in {Kerr} media considering space-time correction}, Phys. Lett. A, 378
  (2014), pp.~1228--1232.

\bibitem{chen_peregrine_2018}
{\sc S.~Chen, Y.~Ye, J.~{Soto-Crespo}, P.~Grelu, and F.~Baronio}, {\em
  Peregrine solitons beyond the threefold limit and their two-soliton
  interactions}, Phys. Rev. Lett., 121 (2018), p.~104101.

\bibitem{fokas_class_1995}
{\sc A.~S. Fokas}, {\em On a class of physically important integrable
  equations}, Phys. D, 87 (1995), pp.~145--150.

\bibitem{fukutani_special_2000}
{\sc S.~Fukutani, K.~Okamoto, and H.~Umemura}, {\em Special polynomials and the
  {Hirota} bilinear relations of the second and the fourth {Painlev\'e}
  equations}, Nagoya Math. J., 159 (2000), pp.~179--200.

\bibitem{gerdjikov_quadratic_1982}
{\sc V.~Gerdjikov and M.~Ivanov}, {\em The quadratic bundle of general form and
  the nonlinear evolution equations}, Bulg. J. Phys., 10 (1983), pp.~130--145.

\bibitem{guo_riemann-hilbert_2012}
{\sc B.~Guo and L.~Ling}, {\em {Riemann-Hilbert} approach and n-soliton formula
  for coupled derivative {Schr\"odinger} equation}, J. Math. Phys., 53 (2012),
  pp.~073506--073506.

\bibitem{kajiwara_determinant_1998}
{\sc K.~Kajiwara and Y.~Ohta}, {\em Determinant structure of the rational
  solutions for the {Painlev\'e IV} equation}, J. Phys. A: Math. Gen., 31
  (1998), pp.~2431--2446.

\bibitem{kametaka_poles_1983}
{\sc Y.~Kametaka}, {\em On poles of the rational solution of the {Toda}
  equation of {Painlev\'e-IV} type}, Proc. Jpn. Acad., Ser. A, 59 (1983),
  pp.~453--455.

\bibitem{kang2018multi}
{\sc Z.-Z. Kang, T.-C. Xia, and X.~Ma}, {\em Multi-soliton solutions for the
  coupled fokas-lenells system via riemann--hilbert approach}, Chin. Phys.
  Lett., 35 (2018), p.~070201.

\bibitem{kaup_exact_1978}
{\sc D.~J. Kaup and A.~C. Newell}, {\em An exact solution for a derivative
  nonlinear {Schr\"odinger} equation}, J. Math. Phys., 19 (1978), pp.~798--801.

\bibitem{lenells_exactly_2008}
{\sc J.~Lenells}, {\em Exactly solvable model for nonlinear pulse propagation
  in optical fibers}, Stud. Appl. Math., 123 (2009), pp.~215--232.

\bibitem{lenells_novel_2009}
{\sc J.~Lenells and A.~S. Fokas}, {\em On a novel integrable generalization of
  the nonlinear {Schr\"odinger} equation}, Nonlinearity, 22 (2009), pp.~11--27.

\bibitem{ling_general_2018}
{\sc L.~Ling, B.-F. Feng, and Z.~Zhu}, {\em General soliton solutions to a
  coupled {Fokas-Lenells} equation}, Nonlinear Anal.: Real World Appl., 40
  (2018), pp.~185--214.

\bibitem{matsuno_direct_2012}
{\sc Y.~Matsuno}, {\em A direct method of solution for the {Fokas-Lenells}
  derivative nonlinear {Schr{\"o}dinger} equation: {II}. dark soliton
  solutions}, J. Phys. A: Math. Theor., 45 (2012), p.~475202.

\bibitem{matveev_darboux_1991}
{\sc V.~B. Matveev, M.~A. Salle, et~al.}, {\em {Darboux} Transformations and
  Solitons}, Springer, 1991.

\bibitem{moses_controllable_2006}
{\sc J.~Moses and F.~W. Wise}, {\em Controllable self-steepening of ultrashort
  pulses in quadratic nonlinear media}, Phys. Rev. Lett., 97 (2006), p.~073903.

\bibitem{okamoto_studies_1986}
{\sc K.~Okamoto}, {\em Studies on the {Painlevé} equations: {III}. second and
  fourth {Painlevé} equations, {P II and P IV}}, Math. Ann., 275 (1986),
  pp.~221--255.

\bibitem{xu_n-order_2012}
{\sc S.~Xu, J.~He, Y.~Cheng, and K.~Porseizan}, {\em The n-order rogue waves of
  {Fokas-Lenells} equation}, Math. Methods Appl. Sci., 38 (2015),
  pp.~1106--1126.

\bibitem{yang_rogue_2021-1}
{\sc B.~Yang and J.~Yang}, {\em Rogue wave patterns in the nonlinear
  {Schr{\"o}dinger} equation}, Phys. D, 419 (2021), p.~132850.

\bibitem{yang_rogue_2022}
\leavevmode\vrule height 2pt depth -1.6pt width 23pt, {\em Rogue wave patterns
  associated with {Okamoto} polynomial hierarchies}, Stud. Appl. Math.,
  (2023), pp.~1--56.

\bibitem{ye_general_2019}
{\sc Y.~Ye, Y.~Zhou, S.~Chen, F.~Baronio, and P.~Grelu}, {\em General rogue
  wave solutions of the coupled {Fokas-Lenells} equations and non-recursive
  {Darboux} transformation}, Proc. R. Soc. A, 475 (2019), p.~20180806.

\bibitem{yue_modulation_2021}
{\sc Y.~Yue and Y.~Chen}, {\em Modulation instability, conservation laws and
  localized waves for the generalized coupled {Fokas-Lenells} equation},
  arXiv:2104.10306,  (2021).

\bibitem{zakharov_exact_1972}
{\sc V.~E. Zakharov and A.~B. Shabat}, {\em Exact theory of two-dimensional
  self-focusing and one-dimensional self-modulation of waves in nonlinear
  media}, Sov. Phys. --- JETP, 34 (1972), pp.~62--69.

\bibitem{zhang_rogue_2022}
{\sc G.~Zhang, P.~Huang, B.-F. Feng, and C.~Wu}, {\em Rogue waves and their
  patterns in the vector nonlinear {Schr\"odinger} equation}, arXiv:2211.05603,
   (2022).

\bibitem{zhang_solitons_2017}
{\sc Y.~Zhang, J.~Yang, K.~Chow, and C.~Wu}, {\em Solitons, breathers and rogue
  waves for the coupled {Fokas-Lenells} system via {Darboux} transformation},
  Nonlinear Anal.: Real World Appl., 33 (2017), pp.~237--252.

\bibitem{zhao_algebro-geometric_2013-1}
{\sc P.~Zhao, E.~Fan, and Y.~Hou}, {\em Algebro-geometric solutions and their
  reductions for the {Fokas-Lenells} hierarchy}, J. Nonlinear Math. Phys., 20
  (2013), pp.~355--393.

\end{thebibliography}

\end{document}